\documentclass[journal]{IEEEtran}
\IEEEoverridecommandlockouts

\makeatletter
\def\markboth#1#2{\def\leftmark{\@IEEEcompsoconly{\sffamily}\MakeUppercase{\protect#1}}%
\def\rightmark{\@IEEEcompsoconly{\sffamily}\MakeUppercase{\protect#2}}}
\makeatother

\usepackage[pagewise]{lineno}
\usepackage[pdftex]{graphicx}
\usepackage{pgfplots}
\def\func1(#1,#2){exp(-(#1-#2)^2/(0.01))}
\usepackage{tikz}
\usepackage{psfrag}
\usetikzlibrary{external}
\tikzset{external/system call={pdflatex \tikzexternalcheckshellescape 
                                        -halt-on-error
                                        -interaction=batchmode 
                                        -jobname "\image" "\texsource"
                                        && inkscape \image.pdf --export-eps=\image.eps --export-ps-level=3}}
\usepackage{verbatim}
\usepackage{amsfonts}
\usepackage{amssymb}
\usepackage[font=small]{caption}
\usepackage{amssymb}
\usepackage{amsmath}
\usepackage{amsthm}
\usepackage{cancel}
\usepackage{epsfig}
\usepackage{varioref}
\usepackage{cite}
\usepackage{blkarray}

  \usepackage[pdftex]{graphicx}
  \usepackage{subfigure}
  \usepackage{epstopdf}
   \graphicspath{{./}}
   \DeclareGraphicsExtensions{.pdf}

\usepackage{xcolor}
\usepackage{enumerate}

\usepackage{stfloats}

\newcommand{\overbw}[1]{\frac{#1}{N^{\frac{1}{2} +\alpha}}}
\newcommand{\setL}{\mathcal{L}}
\newcommand{\setN}{\mathcal{N}}
\newcommand{\setB}{\mathcal{B}}

\newcommand{\setS}{\mathcal{S}}

\newcommand{\U}{\mathbf{U}}

\newcommand{\Ex}[2]{{\textnormal{E}_{#1}\left[#2\right]}}

\newcommand{\CInf}[3]{{\textnormal{I}\left(#1;#2|#3\right)}}
\newcommand{\Inf}[2]{{\textnormal{I}\left(#1;#2\right)}}

\usepackage{pifont}
\theoremstyle{plain}
\newtheorem{remark}{Remark}

   \definecolor{blueH3}{rgb}{0,.5,1}
   \definecolor{blueH2}{rgb}{0,0.25,0.75}
   \definecolor{blueH1}{rgb}{0,0,0.5}   
   \definecolor{grayOldText}{rgb}{.5,.5,.5}
   \definecolor{VCobalt}{HTML}{005682}
   \definecolor{TZTeal}{HTML}{008080}
   \definecolor{KYJade}{HTML}{008151}
   \definecolor{ARust}{HTML}{a10000}
   \definecolor{FFucsia}{HTML}{7000c3}
   
\newcommand{\CASE}[1]{\STATE \textbf{case} #1\textbf{:} \begin{ALC@g}}
\newcommand{\ENDCASE}{\end{ALC@g}}

\newcommand{\DEFAULT}{\STATE \textbf{default:} \begin{ALC@g}}
\newcommand{\ENDDEFAULT}{\end{ALC@g}}
\newcommand{\DEFAULTLINE}[1]{\STATE \textbf{default:} }

\newcounter{MYtempeqncnt}

\theoremstyle{newstyle}
\newtheorem{theorem}{Theorem}
\newtheorem{lemma}{Lemma}
\newtheorem{corollary}{Corollary}
\theoremstyle{plain}
\newtheorem{definition}{Definition}
\usepackage{balance}
\usepackage{setspace}

\usepackage{mathrsfs}

\title{Capacity scaling in a Non-coherent Wideband Massive SIMO Block Fading Channel}

\author{Felipe Gomez-Cuba, ~\IEEEmembership{Member,~IEEE}, Mainak Chowdhury, ~\IEEEmembership{Member,~IEEE}, Alexandros Manolakos, Elza Erkip,~\IEEEmembership{Fellow,~IEEE}, and Andrea J. Goldsmith,~\IEEEmembership{Fellow,~IEEE} \thanks{
    Part of this work was published in the IEEE Information Theory Workshop 2015 \cite{mainak2015wideband}.
    Felipe Gomez-Cuba is with Dipartimento D'engegneria della Informazione, University of Padova, Italy. Email: \texttt{gomezcuba@dei.unipd.it}
    Mainak Chowdhury is with ZaiNar, Inc. Email: \texttt{mainakch@gmail.com}
    Alexandros Manolakos is with Qualcomm Technologies, Inc. - Wireless Research and Development (R\&D). Email: \texttt{amanolak@qti.qualcomm.com}
    Elza Erkip is with NYU Tandon School of Engineering, Brooklyn, NY, USA. Email: \texttt{elza@nyu.edu}.
    Andrea J. Goldsmith is with the Department of Electrical Engineering, Stanford University, Stanford, CA, US Email: \texttt{andreag@stanford.edu}.
This project has received funding from the European Union's Horizon 2020 research and innovation programme under the Marie Sk\l{}odowska-Curie grant agreement No 704837, from ONR grant N000141210063, and from the NSF Center for Science of Information (CSoI) under grant CCF-0939370.}}

\begin{document}

\maketitle

\begin{abstract}
  The scaling of coherent and non-coherent channel capacity is studied
  in a single-input multiple-output (SIMO) block Rayleigh fading
  channel as both the bandwidth and the number of receiver
  antennas go to infinity jointly with the transmit power fixed. 
  The transmitter has no channel state information (CSI), while the
  receiver may have genie-provided CSI (coherent receiver), or the
  channel statistics only (non-coherent receiver).  Our results show
  that if the available bandwidth is smaller than a threshold
  bandwidth which is proportional (up to leading order terms) to the
  square root of the number of antennas, there is no gap between the
  coherent capacity and the non-coherent capacity in terms of capacity
  scaling behavior. On the other hand, when the bandwidth
  is larger than this threshold, there is a capacity scaling gap.  Since achievable rates using
  pilot symbols for channel estimation are subject to the non-coherent capacity
  bound, this work reveals that pilot-assisted coherent receivers in systems 
  with a large number of receive antennas are unable to exploit excess
  spectrum above a given threshold for capacity gain.
\end{abstract}
\begin{IEEEkeywords}
  Massive MIMO, wideband channel, non-coherent Communications, Energy  modulation
\end{IEEEkeywords}

\section{Introduction} \label{sec:introduction}

  Spectral efficiency continues to be an important driver for next
  generation wireless communication systems, including 5G cellular.
  Hence, communication techniques using large antenna arrays and large
  bandwidths have attracted significant attention from both industry
  and academia. 
  Intuitively, increasing the number of antennas and/or the operating
  bandwidth increases the number of signal dimensions in a
  communication system, enabling improved performance. However, an
  accurate knowledge of the wireless channel parameters at the receiver is
  essential to achieving this performance.  Obtaining knowledge of the
  instantaneous channel state information (CSI) at the receiver
  becomes increasingly difficult when \textit{either} the number of
  antennas or the bandwidth is large.  This has motivated much work
  investigating channel capacity under varying degrees of receiver
  CSI, especially in multi-antenna or wideband systems.

  The impact of receiver CSI  on channel capacity may be studied by
  comparing the capacity of a channel under the assumption that the
  CSI is perfectly known at the receiver (coherent capacity) with the
  capacity of a channel where the CSI is unknown
  (non-coherent capacity). The coherent capacity has been well
  characterized for many point-to-point channels, including those
  arising in the context of a multi-antenna or a wideband system
  \cite{goldsmith2005book}. Results for the capacity of the
  non-coherent channel are more elusive. Existing works on
  non-coherent channel capacity include
  \cite{Marzetta1999,1337103,Perera2006,Lozano2008,Sethuraman2009,Jindal2010,7541626},
  which characterize properties of the capacity, and the capacity-achieving
  distribution.  Some of these works also establish upper and lower
  capacity bounds for certain channel models.

  In addition to capacity expressions and bounds, there has been some
  recent work which has characterized the asymptotic behavior of the
  capacity of various noncoherent channels. Two asymptotic regimes
  that are relevant to our analysis in this manuscript include
  characterizations of the non-coherent capacity achieving schemes
  both in asymptotically wideband systems with a fixed number of
  antennas
  \cite{journals/tit/TelatarT00,journals/tit/MedardG02,journals/tit/Verdu02,Medard2005,Zheng2007noncoherent,Ray2007noncoherent,journals/twc/LozanoP12,Ferrante2016,fgomezUnified,Du2017},
  and in narrowband systems with an asymptotically large number of
  antennas
  \cite{journals/ftcit/TulinoV04,marzetta2006much,marzetta2013special,Bjornson2014a,Bjornson2016,zhengoptimal,866662,hoydis2018}.

  A major focus in the study of wideband systems with a finite number
  of antennas is the inherent capacity gap between coherent and
  non-coherent systems under constraints on the signal ``peakiness''
  \cite{journals/tit/TelatarT00,journals/tit/MedardG02,journals/tit/Verdu02,Medard2005,Zheng2007noncoherent,Ray2007noncoherent,journals/twc/LozanoP12,Ferrante2016,fgomezUnified,Du2017}.
  To narrow or close this gap, it is necessary to use either signals
  with unbounded peakiness
  \cite{Medard2005,Zheng2007noncoherent,Ray2007noncoherent} or to reduce
  the bandwidth occupancy of the signal by restricting the transmitted
  signal power to a small subset of the available spectrum
  \cite{journals/twc/LozanoP12,fgomezUnified}.  In either case, even
  though the rate gap is closed with increasing peakiness, the
  non-coherent receivers have a  \textit{second order
    suboptimality} \cite{journals/tit/Verdu02} compared to coherent receivers in terms of capacity
  growth with increasing bandwidth.

  On the other hand, a significant challenge in narrowband systems
  with large antenna arrays (e.g., in multiuser massive Multiple Input Multiple Output (MIMO) systems)
  is channel estimation \cite{marzetta2006much,Bjornson2014a}.  If the receiver has perfect CSI, the degrees of freedom (DoF) in an i.i.d. Rayleigh channel is the minimum of the number of transmit and receive antennas.  This suggests that, if the number of receiver antennas is high and the number of transmit antennas is increased, the DoF, and hence the overall capacity, increases.   
  Without perfect CSI, however, the DoF cannot be increased arbitrarily by having more transmit antennas \cite{866662}. In fact in a non-coherent Rayleigh fading channel with the number of transmit antennas greater than the coherence length, it has been shown that the capacity does not increase at all by adding more transmit antennas \cite{Marzetta1999}. This implies that increasing the number of transmit antennas cannot make the capacity grow indefinitely in the non-coherent channel as in the coherent channel. The capacity can still grow logarithmically with the number of receive antennas in the non-coherent channel \cite{manolakos2014globecom,Chowdhury2014a}. 
  In multiple-user scenarios where the number of users remains constant as the number of antennas scales a similar scaling limitation exists. Motivated by this, in the remainder of the manuscript, we focus on a single user system with a single transmit antenna (point-to-point SIMO channel). With regard to our main result, in the last section of the paper we argue that our main result on capacity scaling for a point-to-point SIMO channel can be extended to MIMO channels and multiuser systems with a finite number of users/transmitters.

  In this work, we characterize the capacity scaling of the
  \textit{block fading wideband massive SIMO channel} and investigate
  the gap in a scaling law sense between the coherent and non-coherent
  capacities of this channel as both the bandwidth and the number of receiver antennas
  jointly scale to infinity. A major difference of our analysis from
  existing wideband coherent vs. non-coherent capacity comparisons summarized above  
  is that we assume both the number of receive antennas and the
  bandwidth are very large. Assuming both dimensions scale to infinity reveals the following
  challenge in applying prior results for scaling across only one of these two dimensions.  
  Prior work on wideband characterizations
  of non-coherent capacity e.g., \cite{journals/tit/MedardG02,
    Medard2005} with a fixed number of antennas assume a low-SNR regime where
   the available transmit power is
  spread across a very wide bandwidth.  On the other hand, 
  characterizations of non-coherent narrowband channel capacity with a large number of
  antennas (e.g., \cite{Chowdhury2014a, manolakos2014globecom}) cannot
  use low-SNR capacity results since the received SNR
  does not decrease with the number of antennas.  Thus, prior results in only one of these asymptotic
  regimes are not immediately applicable to the situation where both
  the bandwidth and the number of antennas scale jointly.

  The main contribution of this manuscript is the derivation of 
  capacity scaling for the block fading wideband massive SIMO channel
  using a non-coherent receiver with knowledge only of the channel statistics.
  To do this, we first derive a capacity upper bound.  Next, we
  describe an achievable scheme based on the use of energy modulation
  \cite{manolakos2014globecom,Chowdhury2014a,7565517,7404014}.  This
  matches the scaling behavior of the upper bound, thereby
  establishing the capacity scaling result.

  Our analysis reveals that, when the bandwidth is smaller than a
  threshold bandwidth which is proportional to the square root of the
  number of antennas, the capacity of the coherent wideband massive
  SIMO channel displays the same scaling as the capacity of the
  non-coherent wideband massive SIMO channel.  Further capacity gain
  in a scaling sense cannot be achieved by increasing the bandwidth
  above this threshold. 

  An important practical implication of our findings arises for
  pilot-assisted coherent receivers under our channel model, as compared to
  coherent receivers with perfect receiver CSI.  Since we consider
  block fading models with any finite blocklength, our non-coherent
  channel capacity scaling analysis is also applicable when a coherent detection
  technique relies on pilot-assisted channel estimation rather than perfect CSI.
  Thus, our analysis suggests that the coherent channel capacity is not necessarily an accurate
  indicator of rate scaling in practical systems
  when more spectrum and a larger number of antennas are available, as in the case of 5G.
  Furthermore, we find numerically that the trends predicted in our capacity
  scaling law characterization may be applicable outside the asymptotic regimes
  of large bandwidth and numbers of antennas. In particular,  our results indicate 
  that the derived capacity scaling laws are applicable for systems with tens to hundreds of antennas 
  and several tens of MHz of bandwidth which is consistent with the design of some existing wireless systems.

  The rest of this paper is structured as follows. Section
  \ref{sec:model} describes the general frequency-selective Rayleigh
  block fading channel model and the different system
  parameters. Section \ref{sec:theorem} describes our main result on capacity
  scaling. Section \ref{sec:general_structure} contains the
  general analysis strategy and outlines the proof of the
  converse. Section \ref{sec:EM} describes a scheme based on EM
  constellations that achieves the capacity scaling. Section
  \ref{sec:examples} provides numerical examples, discussion and extensions of the main
  results in the paper. Finally Section \ref{sec:conclusion} concludes
  the paper.

\subsection{Notation}

We use uppercase letters (A, B, C, \ldots) to refer to random variables, and lowercase letters (a, b, c, \ldots) to refer to their realization. We use script notation $\mathcal{A}, \mathcal{B}, \ldots$ to refer to sets. We use tensor notation (objects with indices) here to refer to collections
of objects. If the index of a tensor object is a set, that refers to the collection of all objects with the corresponding indices belonging to the set. For example $X_\mathcal{A}$ represents a one dimensional vector with coefficients $\{X_a, \forall a\in\mathcal{A}\}$, $x_{\mathcal{A},\mathcal{B}}$ denotes a two dimensional matrix, and so on.

We use $f(N) = \Theta(g(N))$ to denote that there exist constants $c_1,c_2,N_0>0$ such that for all $N>N_0$, $c_1<\frac{f(N)}{g(N)}<c_2$, and $f(N) = o(b(N))$ to denote that $\lim\limits_{N\rightarrow\infty}\frac{f(N)}{g(N)}=0$ \cite{Knuth1976}. We use $f(N) \lessapprox g(N)$ to denote $f(N) \leq g(N) + o(g(N))$. We use $f(N)\dot{=}g(N)$ to denote $\lim_{N\rightarrow\infty}f(N)-g(N)=0$.



\section{System model and assumptions}
\label{sec:model}

We consider a single-antenna transmitter and $N$ antennas at the
receiver.  We assume an independent identically distributed (i.i.d.)
block fading wideband channel model, such that the bandwidth is
divided into a number of independent frequency-flat subchannels.  Such
a channel may arise, for example, in a rich scattering environment
with a high delay spread.  We assume that the delay spread is constant

We assume further that the subchannel gains experience Rayleigh fading that is
i.i.d. across subchannels and receiver antennas.  The validity of this
assumption depends on wireless propagation in the bands of
operation and on the number of receiver antennas.  In sub-6 GHz bands
with a rich scattering environment, this assumption is typically more accurate than in mmWave bands (e.g., 28
GHz, 38 GHz, 60 GHz) where the propagation is more specular. And less
scattering leads to more correlation across fading realizations. Our
analysis does not apply in channels when the i.i.d. assumption is not valid. Some existing work on wideband sparse channels have reported an ``overspreading''
phenomenon similar to ours. We discuss this related work and possible extensions in Section \ref{subsec:general_channel_models}.

  We take the subchannel bandwidth to be the coherence bandwidth $B_c$ which
  is assumed fixed in this manuscript.  We also assume that the
  coherence time is fixed, equal to $L_c$ symbols and that each subchannel experiences
  narrowband block fading \cite{goldsmith2005book}.  The number of subchannels  is $B,$ thus the total bandwidth is $BB_s.$

The channel can be represented on the resource grid depicted in Figure \ref{fig:channelOFDMAb} where we denote the transmitted symbol indices by $\ell \in \{0, \ldots, L_c - 1\} \triangleq \setL,$ the receive antennas by $r \in \{0, \ldots, N-1\} \triangleq \setN,$ and the subchannels by $m \in \{0, \ldots, B-1\}  \triangleq \setB$. For each triplet $(r,m,\ell)\in\setN\times\setB\times\setL$ the channel output $Y_{r, m,\ell}$ is given by

\begin{equation}
\label{eq:channelmodel}
  \begin{split}
    Y_{r, m,\ell}&=H_{r, m} X_{m,\ell}+Z_{r, m,\ell}.\\
  \end{split}
\end{equation}
where $Z_{r,m,\ell}\sim\mathcal{CN}(0,1)$ is Additive White Gaussian Noise (AWGN). $H_{r,m}\sim\mathcal{CN}(0,1)$ is the channel gain, which is i.i.d. for each $r$ and $m$ but remains constant over all $\ell \in \setL,$ and $X_{m,\ell}$ is the channel input, where the set of transmitted symbols $X_{\setB,\setL}$ are subject to the power constraint

\begin{equation}
\label{eq:pconstraint}\Ex{}{\frac{1}{BL_c}\sum_{\ell \in \setL}\sum_{m \in \setB}|X_{m,\ell}|^2} \leq  P,
\end{equation}
where $P$ is the average power available to the single-antenna transmitter. This is the block Rayleigh fading model with coherence block length $L_c$. In this paper we assume $L_c$ and $P$ are constants; furthermore $L_c=\Theta(1)$ and $P = \Theta(1)$ will be used interchangeably since we study the scaling behavior.

\begin{figure}[!t]
  \centering
  \includegraphics[width=0.7\columnwidth]{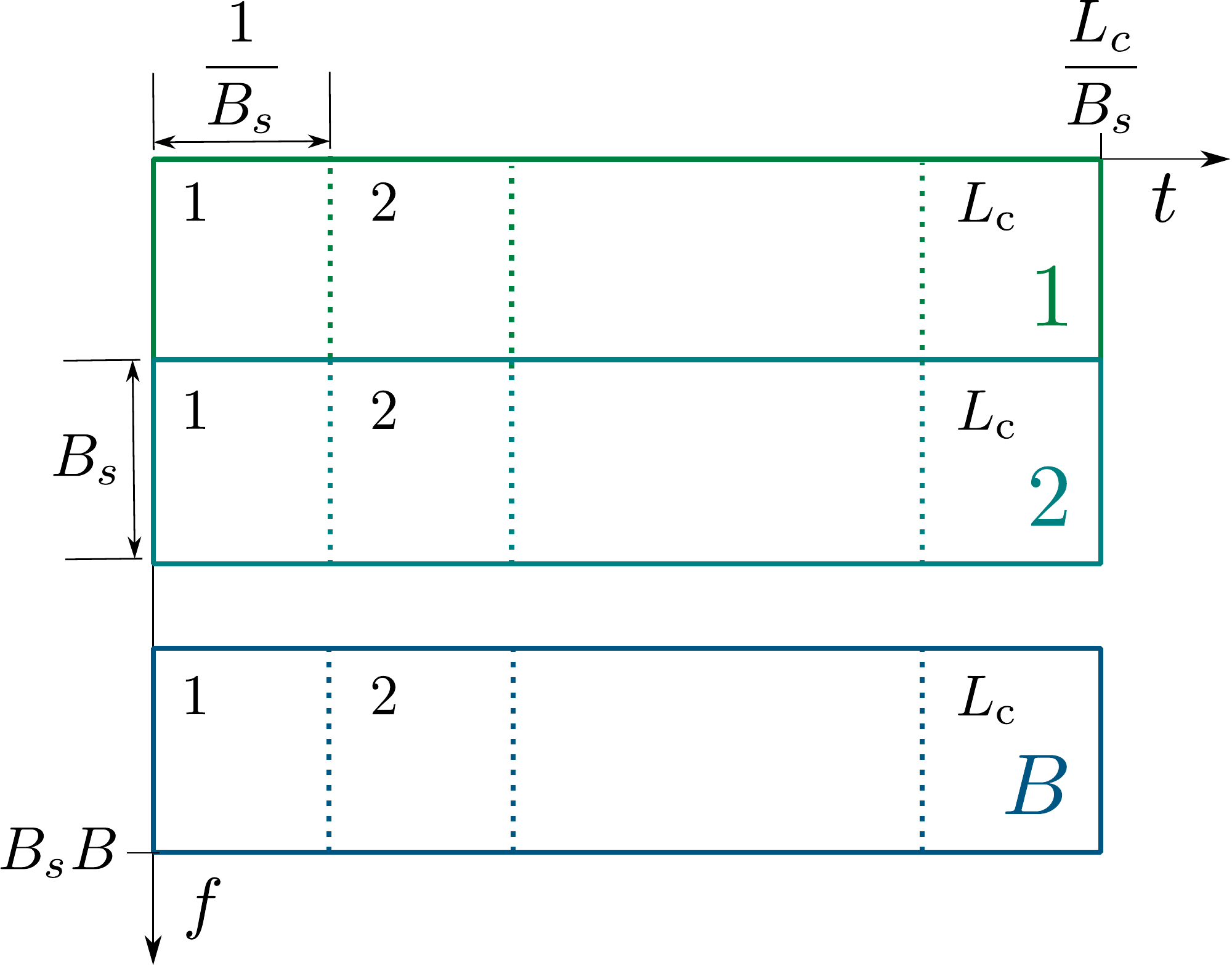}
  \caption{Block-fading frequency-selective channel consisting in
$B$ independent subchannels of bandwidth $B_s$ each, where each subchannel experiences flat i.i.d. Rayleigh block-fading with block length $L_c$.}
  \label{fig:channelOFDMAb}
\end{figure}

In our analysis we compare the ergodic capacity
\cite{goldsmith2005book} of coherent and non-coherent channels as
defined below. Ergodic capacity is a more appropriate performance
indicator than capacity with outage in modern wireless technologies
with turbo coding and hybrid automated repeat request
\cite{Lozano2012}.  Furthermore, we assume no transmitter CSI for both
the coherent and the non-coherent systems.

\begin{definition}
 We define the \textbf{coherent} ergodic capacity of the block fading Rayleigh channel as a function of $N$ and $B$ as

 \begin{equation}
 \label{eq:defCc}
 C_c(N,B)=\sup_{p(X_{\setB,\setL})}\Ex{H_{\setN,\setB}}{\CInf{X_{\setB,\setL}}{Y_{\setN,\setB,\setL}}{H_{\setN,\setB}}}
 \end{equation}
\end{definition}

\begin{definition}
 We define the \textbf{non-coherent} ergodic capacity of the block fading Rayleigh channel as a function of $N$ and $B$ as

 \begin{equation}
 \label{eq:defCn}
 C_n(N,B)=\sup_{p(X_{\setB,\setL})}\Ex{H_{\setN,\setB}}{\Inf{X_{\setB,\setL}}{Y_{\setN,\setB,\setL}}}
 \end{equation}
\end{definition}

Here, note that $C_c(N,B)$ and $C_n(N,B)$ are specified in units of bits per fading block. The coherent capacity can be put in bits per channel use as $\frac{B_s}{L_c}C_c(N,B)$ where $\frac{B_s}{L_c}$ is a constant that does not affect scaling results, and likewise for the non-coherent capacity. Note also that the supremum is outside the expectation operator  and $p(X_{\setB,\setL})$ does not depend on the channel, as there is no feedback nor CSI at the transmitter.

In our analysis, we will study the scaling of $C_c(N,B)$ and $C_n(N,B)$ asymptotically as $B, N$ both go to infinity according to the following relation:

\begin{equation}
  \epsilon = \lim_{B,\,N\rightarrow \infty}
  \frac{\log(B)}{\log(N)},
\label{eq:joint_scaling}
\end{equation}
for a non-negative constant $\epsilon.$  This may equivalently be expressed as $B=\Theta(N^\epsilon).$

\begin{figure*}[!b]
\normalsize
\vspace*{4pt}
\hrulefill
\setcounter{MYtempeqncnt}{\value{equation}}
\setcounter{equation}{9}
\begin{equation}
\label{eq:MIsplit}
 \begin{split}
    C_n(B,N)&=\max_{p(X_{\mathcal{B},\mathcal{L}}):\Ex{}{\frac{1}{BL_c}\sum_{\ell \in \setL}\sum_{m \in \setB}|X_{m,\ell}|^2} \leq  P} \Ex{H_{\setN,\setB}}{\Inf{X_{\mathcal{B},\mathcal{L}}}{Y_{\mathcal{N},\mathcal{B},\mathcal{L}}}}\\
    &=\max_{\rho_{\mathcal{B}}:\sum_{m \in \setB}\rho_m \leq  P} \sum_{m\in\mathcal{B}} \max_{p(X_{m,\mathcal{L}}):\Ex{}{\frac{1}{L_c}\sum_{\ell \in \setL}|X_{m,\ell}|^2} \leq  \rho_m} \Ex{H_{\setN,m}}{\Inf{X_{m,\mathcal{L}}}{Y_{\mathcal{N},m,\mathcal{L}}}}\\
 \end{split}
\end{equation}
\setcounter{equation}{\value{MYtempeqncnt}}
\end{figure*}

\section{Main Result}
\label{sec:theorem}

\subsection{Coherent capacity scaling}

We begin by formalizing the scaling of the capacity of the coherent channel $C_c(B,N)$ for large $N$. Since optimal inputs and receiver combining are well known for the SIMO channel with perfect CSI, the scaling result follows from simple observations on known expressions.

\begin{lemma}
\label{ref:lemCoh}
For $\epsilon < 1,$  the capacity of a coherent block Rayleigh fading i.i.d. wideband SIMO channel with $N$ receive antennas and bandwidth $B=\Theta(N^\epsilon)$, where $H_{\setN,\setB}$ is known perfectly at the receiver, and the transmitter has no CSI nor feedback, scales as

 \begin{equation}C_c(B,N)=\Theta(N^{\epsilon} \log(N)).\end{equation}
 For $\epsilon\geq1$ it scales as

 \begin{equation}C_c(B,N)=\Theta(N).\end{equation}
\end{lemma}
\begin{proof}
  The SIMO channel capacity with perfect CSI at the receiver is well known. We consider a SIMO model where for each pair $(m,\ell)\in\setB\times\setL$ the subchannel $H_{\setN,m}$ is an independent vector. Since there is no feedback or CSI at the transmitter, the ergodic capacity is achieved using Gaussian inputs with equal power in all subchannels \cite{goldsmith2005book}. Since there is only one transmit antenna, maximal ratio combining can be used at the receiver and the capacity is
  \begin{equation}
    \label{eq:coherent_cap_erg}
    C_c(N,B)=\Ex{H_{\setN, m}}{B B_s\log(1+\frac{P (\sum_{r=0}^{N-1}|H_{r, m}|^2)}{B })} \,\, \text{bit/s}.
  \end{equation}
The expectation is taken over the randomness in the subchannel coefficients. $\sum_{r=0}^{N-1}|H_{r, m}|^2$ is chi-square distributed with $2N$ degrees of freedom and in the limit $N \rightarrow \infty$ the term $\frac{P (\sum_{r=0}^{N-1}|H_{r, m}|^2)}{B }\to\frac{NP}{B}$ concentrates sharply around a quantity proportional to $N^{1-\epsilon}, $ since $P$ is constant.  Thus, with $\epsilon< 1$, the ergodic capacity scales as $\Theta(B\log(N)).$ For $\epsilon=1,$  the SNR per subchannel does not increase, hence the ergodic capacity scales as $\Theta(B) = \Theta(N)$.  Finally, for $\epsilon> 1$ the capacity expression is the same but the effective SNR per subchannel decreases as $N$ increases. Using the approximation $\log(1+x) \simeq x$ for $x\ll1$ the ergodic capacity is proportional to $\frac{N^{1+\epsilon}}{N^\epsilon}$ for sufficiently large $N$, which reduces to $\Theta(N)$.
\end{proof}

\subsection{Non-coherent capacity scaling}

A closed form expression for the capacity of the general multi-antenna non-coherent channel is not available. Instead, we characterize the scaling behavior of an upper bound to the capacity and show that a scheme with the same rate scaling exponent is achievable. The upper bound and achievable scheme coincide in a scaling law sense but the corresponding rates may differ by constant or lower order terms.

\begin{theorem}
\label{th:capacity}
The capacity scaling for the non-coherent block fading i.i.d. Rayleigh wideband SIMO channels with $N$ receive antennas and bandwidth $B=N^\epsilon,$ with a finite coherence time $L_c,$ satisfies

\begin{equation}\Theta(N^{\min(\epsilon, \frac{1}{2}) - \alpha})\leq C_n(B, N) \leq \Theta(N^{\min(\epsilon, \frac{1}{2}) + \alpha})\end{equation} for all $\alpha > 0$.
\end{theorem}
\begin{proof}
  The proof is trivial for $\epsilon < \frac{1}{2},$ for
  which we use the coherent capacity bound, i.e.,\[C_n(B, N)\leq C_c(B,
  N).\] Therefore with $\epsilon<\frac{1}{2}$ the upper bound is
  $C_n(B, N) \leq \Theta(N^{\epsilon})$. In addition,  the
  EM scheme discussed in Section \ref{sec:EM} achieves a rate
  $\Theta(N^\epsilon)$ with a vanishing error probability if and only
  if $\epsilon < \frac{1}{2}$. In addition, for the case
  $\epsilon>\frac{1}{2}$ it is trivial to construct an achievable
  scheme by using a subset of the channels, i.e., selecting
  $B'=\sqrt{N}<B$ subchannels to transmit using the same EM scheme
  described in Section \ref{sec:EM}. Hence, the remainder of the proof
  consists in showing that when $\epsilon > \frac{1}{2}$ a rate
  scaling $\Theta(N^\epsilon)$ is not achievable when the channel
  realization is not known at the receiver.  In a prior work
  \cite{mainak2015wideband}, we derived a bound on the same exponent
  for the achievable rate with inputs subject to a PAPR (peak to
  average power ratio) constraint with the fading blocklength $L_c=1$.
  In the current  work, we do not make any assumptions on the PAPR or
  $L_c, $ except for the  latter staying constant.

We establish a  capacity
upper bound for $\epsilon > \frac{1}{2}$ in Section \ref{sec:general_structure}.  This reveals that any capacity optimal scheme cannot use bandwidth more than $\Theta(N^{\frac{1}{2}}).$   We then show in  Section \ref{sec:EM}, an achievability result using energy modulation which matches the capacity scaling in the sense of the theorem above.  This establishes the theorem.

\end{proof}

Comparing Lemma \ref{ref:lemCoh} and Theorem \ref{th:capacity}, we
observe that the capacities of coherent and non-coherent channels
scale the same way when the bandwidth scales slower than a certain
threshold ($\Theta(N^{1/2})$), whereas beyond that threshold,
additional bandwidth does not help increase the non-coherent capacity scaling.



Furthermore, we have the following corollary:
\begin{corollary}
  The capacity scaling of coherent and non-coherent receivers is the same in a wideband massive SIMO system if the bandwidth is lower than a threshold bandwidth proportional to the square root of the number of antennas.
\end{corollary}

\begin{proof}
  This follows directly from Lemma \ref{ref:lemCoh} and Theorem \ref{th:capacity}.
\end{proof}

\begin{remark}
  There is a rate scaling gap between genie-aided coherent receivers and pilot-assisted coherent receivers when the bandwidth scales faster than $\Theta(N^{1/2})$ in a wideband block Rayleigh fading SIMO channel.
\end{remark}

This remark follows from that fact that coherent receivers considered in this manuscript use genie-aided perfect
  noiseless CSI, which involves knowing $H_{\setN,\setB}$ perfectly.
  CSI obtained in any practical wireless system is not perfect, it is
  always acquired through the use of pilot symbols in a noisy channel.
  Thus pilot-assisted coherent receivers which perform
  channel estimation and symbol demodulation separately are also
  subject to the non-coherent capacity limit derived in this
  manuscript.  

These implications will be illustrated through simulation results in Section \ref{subsec:num_plots}.

\section{Proof of the converse}
\label{sec:general_structure}

In this section, we derive an upper bound on the capacity of the SIMO system per subchannel \eqref{eq:MIsplit}.  We then  exploit  the fact that the wideband SIMO channel consists of multiple
parallel subchannels which are statistically identical to derive a final upper bound on the capacity.

\subsection{Determining the number of active subchannels}
\label{sec:reductionbandwidthallocation}

Due to the fact that the channel coefficients $H_{r,m}$ in \eqref{eq:channelmodel} are independent and the transmitter has no CSI or feedback, the mutual information maximization problem in \eqref{eq:MIsplit}  can be expressed as the separate maximization problems over a collection of parallel subchannels coupled together only by the average input power constraint.

\stepcounter{equation}

We observe that when $\rho_m $ is fixed, the maximization problems over the subchannels are decoupled.  We can define the maximum mutual information achievable on subchannel $m$ as a function of the subchannel transmit power $\rho_m$ as follows:

\begin{equation}
\label{eq:MIfun}
f_{\mathrm{C},m}(\rho_m)\triangleq\max_{\substack{p(X_{m,\mathcal{L}}):\\\Ex{}{\frac{\sum_{\ell \in \setL}|X_{m,\ell}|^2}{L_c}}} \leq  \rho_m}\Ex{H_{\setN,m}}{\Inf{X_{m,\mathcal{L}}}{Y_{\mathcal{N},m,\mathcal{L}}}}.
\end{equation}

Since the $H_{r,m}$ in \eqref{eq:channelmodel} are i.i.d., the function $f_{\mathrm{C},m}(.)$ is the same for all subchannels. We denote this by dropping the index $m$. This allows us to rewrite the capacity of the wideband channel as the following sum of identical subchannel capacity functions.

\begin{equation}C_n(B,N)=\max_{\rho_{\mathcal{B}}:\sum_{m \in \setB}\rho_m \leq  P}\sum_{m\in\mathcal{B}}f_{\mathrm{C}}(\rho_m)
\end{equation}

We observe that $f_{\mathrm{C}}(\rho_m)$ is a non-decreasing function of $\rho_m$.  By the symmetry of the problem and the KKT conditions, the optimal capacity achieving solution consists of selecting a certain number of subchannels and dividing the total transmit power equally across them.

This means that the non-coherent capacity of the wideband channel can be written as

\begin{equation}
\label{eq:MIwrap}
 \begin{split}
    C_n(B,N)&=\max_{M\in\{1\dots  B\}} M f_{\mathrm{C}}\left(\frac{P}{M}\right)\\
 \end{split}
\end{equation}
where, to fully characterize capacity, we would have to address two optimizations: first, find the optimal per-subchannel input distribution in \eqref{eq:MIfun} when $\rho_m=\frac{P}{M}$; and second, find the optimal number of subchannels that must be used in order to maximize \eqref{eq:MIwrap}, denoted by $M^*$.

\subsection{Separation of Energy and Shape Distributions}
\label{sec:inputenergyseparation}

The following lemma is Theorem 2 in \cite{Marzetta1999}, particularized from MIMO to SIMO and written using our notation. We will use this lemma to simplify the inner optimization contained in \eqref{eq:MIfun}.
\begin{lemma}
\label{lem:marzetta}
 \textbf{(Theorem 2 of [3])} The optimal distribution for a frequency-flat Rayleigh block fading $1\times N$ SIMO (sub)channel is $X_{m,\mathcal{L}}=\sqrt{A_m} U_{m,\mathcal{L}}$ where $A_m$ and $U_{m,\mathcal{L}}$ are independent, $A_m$ is a non-negative real number distribution and $U_{m,\mathcal{L}}$ is an Isotropically Distributed Unitary Vector (IDUV).
\end{lemma}

Using this result the input optimization problem in \eqref{eq:MIfun} reduces as follows: given $X_{m,\mathcal{L}}=\sqrt{A_m} U_{m,\mathcal{L}}$ with known optimal $p(U_{m,\mathcal{L}})$, we need only to optimize $\Inf{\sqrt{A_m}U_{m,\mathcal{L}}}{Y_{\mathcal{N},m,\mathcal{L}}}$ over $\displaystyle \{p(A_{m}):\Ex{}{A_{m}}=\rho_m\}$. Still, the amplitude distribution optimization over $p(A_{m})$ contained in $f_{\mathrm{C}}(\rho_m)$ to characterize capacity exactly remains elusive. Instead, in Sections \ref{sec:UBenergy} and  \ref{sec:UBenergy_greater_lc} we obtain scaling upper bounds on $f_{\mathrm{C}}(\frac{P}{M})$ as a function of $M$ and based on these results we identify a limit to the scaling of the optimal number of subchannels $M^*$ and obtain an upper bound on capacity scaling, i.e. on the scaling of $M^*f_{\mathrm{C}}(\frac{P}{M^*})$. In addition, we present an EM scheme in Section \ref{sec:EM} that achieves the capacity scaling upper bounds in all regimes and thus shows our result fully characterizes the capacity scaling.

This generalizes our prior result in \cite{mainak2015wideband} where we had employed a different proof that required two assumptions: that the fading blocklength was one ($L_c=1$) and that the input distribution $p(X_{m,\mathcal{L}})$ was subject to a peak power constraint. For ease of presentation, we introduce our extension by first removing one constraint, the peak power constraint, while keeping the other, $L_c=1$ in Section \ref{sec:UBenergy}, and next we generalize the proof to $L_c>1$ in Section \ref{sec:UBenergy_greater_lc}.

\subsection{Scaling upper bound for $C_n(B,N)$ with $L_c = 1$, $\epsilon>\frac{1}{2}$}
\label{sec:UBenergy}

Taking Lemma \ref{lem:marzetta} (i.e. Theorem 2 of \cite{Marzetta1999}), $L_c=1$  implies that  the isotropic vector $U_{m,\setL}$ becomes a random phase that cannot be recovered by the receiver, and the rate depends completely on the information carried by the distribution of the input energy, i.e., $p(A_m)$ with $A_m=|X_{m, 0}|^2.$ In this subsection, we assume that the transmit power is spread equally across $M$ ``active'' subchannels and investigate if the total capacity, i.e.,  $Mf_{\mathrm{C}}(P/M),$ can exceed $\Theta(N^{\frac{1}{2} + \xi})$ for any positive $\xi.$ We focus on the case when the bandwidth scales like $M=\Theta(N^{\frac{1}{2} +  \alpha})$ for $\alpha > 0,$ since when the bandwidth scales like $\Theta(N^{\frac{1}{2} - \alpha})$ for any positive $\alpha$, our achievable scheme achieves the scaling behavior of the coherent channel capacity scaling.  With $M=\Theta(N^{\frac{1}{2} +  \alpha})$ we have the following upper bound for the capacity function in each active subchannel.

\begin{lemma}
    \label{lemm:opt_mi}
    When $M=\Theta(N^{\frac{1}{2} +\alpha})$ with $\alpha>0$, $\rho_m=\frac{P}{M}$ and $L_c=1$, the subchannel capacity function \eqref{eq:MIfun} satisfies

  \begin{equation}f_{\mathrm{C}}(P/M) \leq  \Theta\left(\frac{1}{N^{2\alpha}}\right).
  \end{equation}
\end{lemma}

\begin{proof}
  Appendix \ref{app:lemm:separation} provides the proof for this. The proof relies on the observation that when two values of input energies $a_m,a_m'\sim A_m$ are close to each other, the output distributions $p(Y_{\setN,m,\setL}|a_m)$ and $p(Y_{\setN,m,\setL}|a_m')$ are indistinguishable in a sense made precise in Appendix \ref{app:lemm:separation}. Moreover, the input energies cannot be too far away from zero without violating the average energy constraint.
\end{proof}

\begin{remark}
\label{rem:opt:mi}
Lemma \ref{lemm:opt_mi} implies that when there are $M=\Theta(N^{\frac{1}{2} + \alpha})$ ``active'' subchannels, the total mutual information in the wideband channel is given by

\begin{equation}Mf_{\mathrm{C}}(P/M) \leq \Theta(N^{\frac{1}{2} - \alpha}).
\end{equation}
Therefore for any $\alpha>0$ we have $Mf_{\mathrm{C}}(P/M)=\Theta(N^{\frac{1}{2} - \alpha})\leq \sqrt{N}f(P/\sqrt{N})=\Theta(N^{\frac{1}{2}})$. Therefore the optimal allocation $M^*$ that maximizes \eqref{eq:MIwrap} satisfies $M^*\leq\Theta(N^{\frac{1}{2}})$. Since $f_{\mathrm{C}}(P/M^*)$ decreases with $M^*$, $M^*f_{\mathrm{C}}(P/M^*)\leq \Theta(M^*)\leq\Theta(N^{\frac{1}{2}})$, and thus the converse of Theorem \ref{th:capacity} with $\epsilon>\frac{1}{2}$ is proven for $L_c=1$
\end{remark}

\subsection{Upper bound on non-coherent capacity scaling, $L_c > 1$}
\label{sec:UBenergy_greater_lc}

For the case $L_c>1$, the Marzetta-Hochwald optimal input result in Lemma \ref{lem:marzetta} \cite{Marzetta1999} gives an input with the isotropic vector $U_{m,\setL}$ with more than one component. In the case $L_c=1$ no information about $U_{m,\setL}$ could be recovered from the channel output, but for $L_c>1$ the channel output may give information about $U_{m,\setL}$ in the sense that some phase-differences and amplitude-differences between $U_{m,\ell}$ and $U_{m,\ell+1}$ might be recoverable. Therefore, to construct the extension of  Lemma \ref{lemm:opt_mi}, the main difference from the proof when $L_c = 1$ is we must show that, when the sequence length $L_c$ is a constant that does not scale with $N$, instead of ``all the rate'' , now ``most of the rate'' depends on the information carried by the distribution of energy of the input.  The input energy density is now defined as $p(A_m)$ where the energy satisfies $A_m=\sum_{\ell = 0}^{L_c -1}|X_{m, \ell}|^2$. The exact sense in which ``most of the rate'' is interpreted is made clear in the proof.

We state the new lemma for the mutual information per subchannel as follows
\begin{lemma}
    \label{lemm:opt_miLc}
    When $M=\Theta(N^{\frac{1}{2} +\alpha})$, with $\alpha>0$, $\rho_m=\frac{P}{M}$ and $L_c\geq 1$ remains constant as $N$ grows, the subchannel capacity function \eqref{eq:MIfun} satisfies

\begin{equation}f_{\mathrm{C}}(P/M) \leq  \Theta\left(\frac{1}{N^{2\alpha}}\right).
\end{equation}
\end{lemma}
\begin{proof}
Details are presented in Appendix \ref{app:lemm:separation_block}.
\end{proof}

\begin{remark}
Although the proof requires additional steps, the scaling bound in Lemma \ref{lemm:opt_miLc} takes the same value as in Lemma \ref{lemm:opt_mi}. Therefore Remark \ref{rem:opt:mi} applies also to the general case $L_c\geq1$.
\end{remark}

\section{Achievability using Energy Modulation}
\label{sec:EM}

We consider the wideband counterpart of the EM scheme initially described in \cite{Chowdhury2014a, manolakos2014globecom} for a narrowband non-coherent massive SIMO system with $M=1$ and $L_c=1$. The narrowband EM scheme modulates information only in the amplitude of the input using a non-negative-valued energy input constellation, transmitting $X\in\mathbb{R}^+\cup\{0\}$ such that $X^2\in \mathcal{C}$. Thus the receiver does not utilize the phase of the output and computes a quadratic total energy statistic across all receive antennas, $V=\frac{1}{N}\sum_{\setN} |Y_{r}|^2$. As $N\to\infty$ with constant power the received statistic satisfies $V\to X^2+1 \pm o(1/N)$ where the second term stems from the noise power $\frac{1}{N}\sum_{\setN}|Z_{r}|^2\to 1$. Thus, the receiver can decode by unequivocally mapping different received-energy regions over $V$ to different inputs $X$.

In this section we define the EM scheme for a wideband system with any arbitrary Rayleigh fading block length constant $L_c\geq1$. We consider the case of transmitting independent streams of information in $M$ subchannels, with the power equally spread across all subchannels. In relation to our upper bounds in Section \ref{sec:general_structure} note that $M$ is the number of active subchannels that may be selected as less than or equal to the number of available subchannels $B=\Theta(N^\epsilon)$. Our main challenge is to characterize the joint error probability across $M$ parallel EM transmissions when the power of each is $\Theta(1/M)$ instead of constant. We show that even in this case,  $\forall \alpha>0$ if $M\leq\Theta(N^{\frac{1}{2}-\alpha})$ the error probability still vanishes as $N\to\infty$ and EM can achieve an arbitrarily low probability of error when we select a constellation $\mathcal{C}$ of a specific type such that the rate is $M\log_2|\mathcal{C}|=\Theta(M\log(N))$.

To design the wideband EM achievable scheme we assume the transmitter wishes to transmit a total rate of $R$ bits per coherence block (can be converted to $R\frac{B_s}{L_c}$ bits per second where $B_s$ and $L_c$ are constants). In each band, the transmitter uses the same constellation $\mathcal{C}$ and only one point in the energy constellation is transmitted during the $L_c$ symbols of duration of the fading block. Hence, to transmit a total of $R$ bits the constellation size should satisfy $|\mathcal{C}| = 2^{R/M}, $ and to satisfy the power constraint the average power of the constellation is $\frac{P}{M}$. Without loss of generality let us normalize the transmitted power $P=1$.

We design the wideband EM scheme for arbitrary $L_c$ assuming an energy detector that removes the need to perform encoding across multiple symbols of the same block ($\ell$). The transmitter sends $X_{m,\ell}=\frac{\sqrt{A_m}}{\sqrt{L_c}}$ where $A_m \in \mathcal{C}$. In terms of the optimal capacity-achieving input distribution in Lemma \ref{lem:marzetta}, $X_{\ell,m}=\sqrt{A_{m}}U_{m,\setL}$, the EM scheme replaces the IDUV $U_{m,\setL}$ by the all-ones normalized vector $\frac{1_{\setL}}{\sqrt{L_c}}$. Since we showed in Lemma \ref{lemm:opt_miLc} that when $L_c$ is a constant ``most of the rate'' depends on the information carried by the distribution of energy, the EM scheme pays only a constant rate penalty  for not using the optimal $U_{m,\setL}$, and the achievable rate with this scheme achieves the scaling exponent of the capacity upper bound.

In the wideband EM scheme the receiver computes the average energy quadratic statistic in each subband $m$ as follows:

  \begin{equation}V_{m}=\frac{1}{NL_c}\sum_{r=0}^{N-1}\left|\sum_{\ell=0}^{L_c-1}Y_{r,m,\ell}\right|^2.
  \end{equation}
Based on its knowledge of the statistics of the channel and of
$\mathcal{C}$, the receiver divides the positive real line into
non-intersecting intervals or \emph{decoding regions}
$\{\mathcal{V}_k\}_{k=1}^{|\mathcal{C}|}$ and returns $ \hat {k} \in
\left\{ \tilde{k}: V_m \in \mathcal{V}_{\tilde{k}} \right\}.$

\begin{theorem}
  \label{thm:achieve}
  In a non-coherent block fading i.i.d. Rayleigh wideband SIMO channel with $N$ receive antennas, using equal power and rate allocation over $M=\Theta(N^{\epsilon})$ subchannels with a finite coherence time $L_c$, for any $\epsilon<\frac{1}{2}$, EM achieves a vanishing probability of error with increasing $N$ using a constellation with transmitted rate that scales as $R=\Theta\left(N^{\epsilon}\log(N)\right)$.
\end{theorem}

\begin{proof}
   We assume the transmitter chooses the constellation $\mathcal{C}$ as

  \begin{equation}
    \begin{array}{ll}
      \mathcal{C} = \left \{ 0, 2d, 4d,\cdots,\frac{2}{M} \right \},
    \end{array}
  \end{equation}
  where we choose the half-distance between symbols as

  \begin{equation}
    \begin{array}{ll}
      \label{eq:db}
      d = \frac{1}{M(|\mathcal{C}|-1)} = \frac{1}{M\left(2^{\frac{R}{M}}-1\right)},
    \end{array}
  \end{equation}
  where $R$ is the transmission rate. We denote the scaling of this distance as $d =\Theta\left( \frac{1}{N^{t}}\right)$ with exponent $t$.
  The transmitter may choose any $d$ satisfying $\epsilon<t<\frac{1}{2}$, where the lower limit is required to
  achieve the specified transmission rate scaling and the upper limit is required by the error analysis below.

  For this scheme, the transmitted rate $R$ must have the scaling $\Theta\left(N^{\epsilon}\log(N)\right)$ as we
  have assumed in the theorem, but this rate must also satisfy an equality relative to the number of subchannels times the bits per
  symbol of the constellation: $R=M\log_2|\mathcal{C}|$. Therefore $R$ depends on $d$
  via the cardinality of the constellation $|\mathcal{C}| = \Theta(N^{t-\epsilon})$.
  We can show both requirements of $R$ are compatible by writing the scaling of $R=M\log_2|\mathcal{C}|$ as follows

  \begin{equation}
   \begin{array}{ll}
    \label{eq:achievable}
    R &= \Theta\left(N^{\epsilon}\log_2(1+N^{t-\epsilon}) \right) \\
    &=  \Theta\left(N^{\epsilon}\log_2(N) \right)
   \end{array}
  \end{equation}
  which agrees with the assumption of the theorem and shows that the chosen exponent $t\geq\epsilon$
  affects the rate as a multiplying constant without changing its scaling with $N$. We next
  show that we can achieve a vanishing error probability as long as we choose
  $t<\frac{1}{2}$, i.e. when the chosen separation between nearest constellation
  points $d=\Theta(N^{-t})$ decays slower than $\Theta(N^{-\frac{1}{2}})$.

  Since $V_m$ approaches $1+A_m$ the centers of the decoding regions are a shifted version
  of the energy constellation symbols (shifted by the noise power), defined as
  $\mathcal{V}_{1} = \left(-\infty, 1+d \right],$ $\mathcal{V}_{k} = \left((2k-1)d +1, (2k+1)d +1\right]$
  for $2\leq k \leq |\mathcal{C}|-1,$ and $\mathcal{V}_{|\mathcal{C}|} =  \left( (2|\mathcal{C}|-1)d+1, \infty\right)$.

  For these decoding regions, we consider the union bound on the probability of error over all the subchannels as follows

  \begin{align}
    \label{eq:ub_prob_error}
    \begin{split}
      P_{\text{error}} &\leq  \sum_{m=1}^{M} \sum_{k \in \mathcal{C}} P_{\text{error},m} (\hat{k} \neq k) \\
      &\overset{(a)}{\approx} N^{\epsilon}
      N^{t-\epsilon} e^{-N (\min_k I_k(d))}\\
      &\overset{(b)}{\approx} N^{t}e^{-N^{1-2 t}},
    \end{split}
  \end{align}
  where $P_{\text{error},m}$ is the probability of error due to
  transmission in the $m^{th}$ subband, and $(a)$ and $(b)$ follow from the narrowband EM analysis in \cite{7404014}. Particularly,  step $(a)$ in \eqref{eq:ub_prob_error} upper bounds $P_{\text{error},m}$ using the error exponent function defined as

  \begin{equation}
    \begin{array}{ll}
      I_k(d) 
      \triangleq \lim\limits_{N \rightarrow \infty} \frac{-\log \left( \text{Prob} \left( \left|V_m-2(k-1)d-1\right| > d \right)\right)  }{N}.
    \end{array}
  \end{equation}
  In EM the Maximum Likelihood detector becomes the energy-based detector for the quadratic statistic $V_m$ \cite[Lemma 1]{7404014}. Finally step $(b)$ in \eqref{eq:ub_prob_error} follows from the limit
  \begin{equation}\lim_{d\rightarrow 0} \frac{I_k(d)}{d^2} =
  \Theta(1)
  \end{equation} for all $k$, which is proven in \cite[Lemma 3]{7404014}.

  Therefore, for $\epsilon<\frac{1}{2}$ the transmitter can choose from a family of constellations satisfying $\epsilon<t<\frac{1}{2}$ such that the rate is $\Theta(N^\epsilon\log(N))$ and the right hand side expression in \eqref{eq:ub_prob_error} goes to zero as $N$ grows, completing the proof of Theorem \ref{thm:achieve}.
\end{proof}

\begin{remark}
  Note that the proof of Theorem \ref{thm:achieve} offers a collection of achievable schemes that spans a family of constellations characterized by the parameter $t$ with $\epsilon<t<\frac{1}{2}$. Any constellation in this family achieves the specified rate scaling with vanishing error probability. Different constellations with different parameters $t$ are differentiated by a trade-off between their rate and error slope, in the sense that a greater $t$ does not modify rate scaling but increases a non-scaling constant multiplier of the achieved rate. Moreover a larger parameter $t$ decreases the exponent that characterizes the decay speed of the error probability as $N$ grows. Finally, the closer $\epsilon$ is to $\frac{1}{2}$ the narrower the range of valid constellations and for $\epsilon=\frac{1}{2}$ the error probability union bound becomes a constant and does not vanish.
 \end{remark}


\section{Numeric Examples, Discussion and Extensions}
\label{sec:examples}

\subsection{Intuition behind capacity scaling}
\label{sec:interpretation_proofs}
 Fig. \ref{fig:comparison} provides an intuitive interpretation of Theorems \ref{th:capacity} and \ref{thm:achieve}. The values of the output energy statistic conditioned on the input energy $p(V_m|A_m=a_m)$ are distributed around $a_m+1$ with most of the probability measure concentrated in a region of size $1/N^{\frac{1}{2}}$ around the center. In Fig. \ref{fig:bwlimited} the inputs are separated as $|a_m'-a_m|=\frac{1}{N^{\frac{1}{2}-\alpha}},$ for $a_m',a_m\in\mathcal{C}:a_m'\neq a_m$, the centers become closer at a slower pace than the tails shrink, and the transmitted symbols become easier to distinguish as $N$ increases (Theorem \ref{thm:achieve}). Conversely, in Fig. \ref{fig:pwoverspread} we have $|a_m'-a_m|=\frac{1}{N^{\frac{1}{2}+\alpha}},$ and, as $N$ grows, the centers of the two conditional output distributions become closer faster than their tails vanish, which is consistent with the interpretation of the proof of Lemma \ref{lemm:opt_mi} discussed in Section \ref{sec:UBenergy}.

  \begin{figure}
    \centering
\subfigure[Bandwidth-limited regime: Increasing bandwidth helps increase capacity]{
    \includegraphics[width=.7\columnwidth]{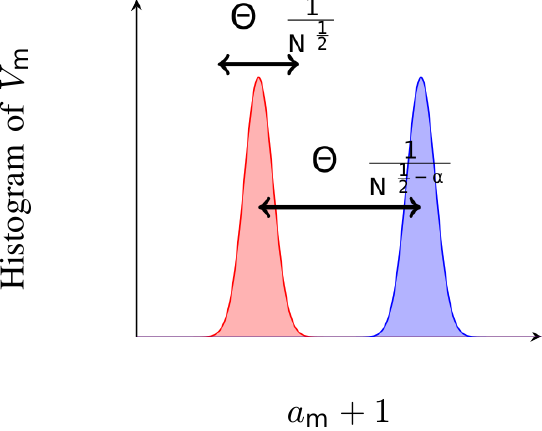}
      \label{fig:bwlimited}
}
\hspace{.1in}
\subfigure[Excessive bandwidth regime: Power is spread over an excessively large bandwidth]{
    \includegraphics[width=.7\columnwidth]{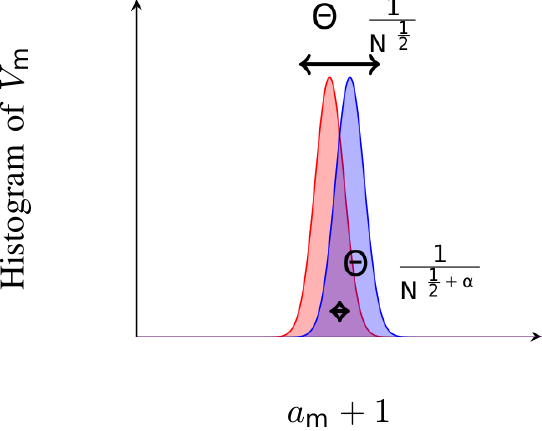}
      \label{fig:pwoverspread}
}
    \caption{Bandwidth-limited and excessive bandwidth capacity regimes}
    \label{fig:comparison}
  \end{figure}

  \subsection{Performance obtained with achievable scheme}

In Figure \ref{fig:EMperformance} we illustrate the Bit Error Rate (BER) and the rate for a series of EM numerical simulation examples. All simulated points in the figure are obtained by the Monte-Carlo method with $10^5$ bits per point. As a rule of thumb this number of Monte-Carlo simulations is sufficient to estimate BERs up to $10^{-4}$, however to conduct simulations both $B$ and $N$ must be integers, so we rounded the number of subchannels as $M=\lceil N^\epsilon \rceil$. We remark that the curves are not smooth due to the bumpy behavior of the ceiling function and not due to a lack of sufficient bits in simulation.

In Figure \ref{fig:EMperformanceBER} we observe that the BER vanishes
as $N$ grows for schemes with $\epsilon<\frac{1}{2}$, whereas the
error does not decay when $\epsilon>\frac{1}{2}$. Moreover, the error
vanishes more slowly when $\epsilon$ is close to the threshold, as
expected, by the bounds in the proof of Theorem \ref{thm:achieve}. In
Figure \ref{fig:EMperformanceRate} we see the transmitted rate of the
EM scheme for each scenario, which is simply computed as
$M\log |\mathcal{C}|$. The rates for all EM schemes grow consistently
with $M=\lceil N^\epsilon \rceil$, yet the rates with exponents
greater than $\frac{1}{2}$ cannot be sustained with an arbitrarily low
error probability for large enough $N$. Here the rates are given in
units of bits per channel coherence block, so for example if
$L_c/B_s=1$ ms then the unit is Kb/s.

\begin{figure}
\centering
\subfigure[BER]{
 \includegraphics[width=0.85\columnwidth]{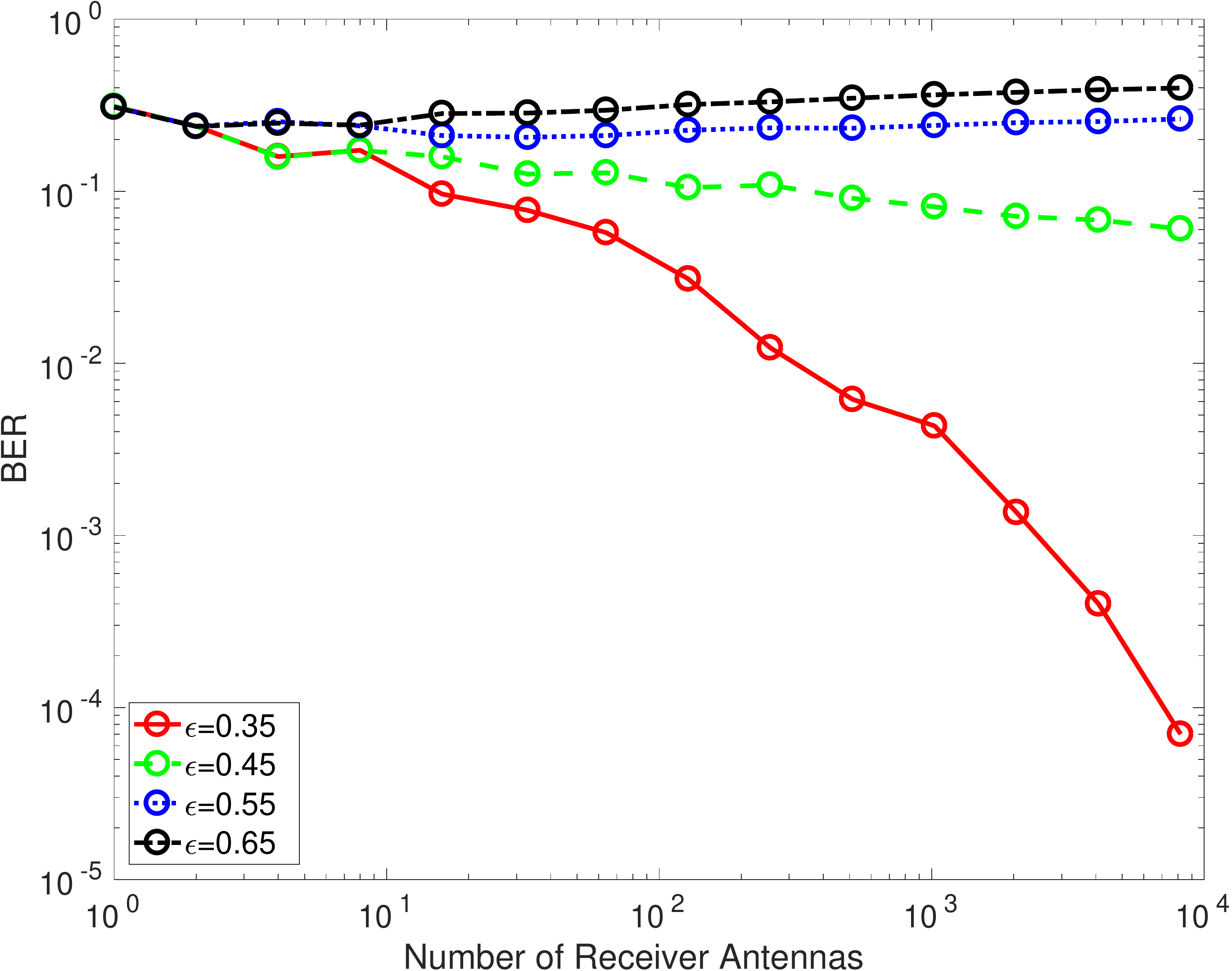}
 \label{fig:EMperformanceBER}
}
\subfigure[Transmitted Rate]{
 \includegraphics[width=0.85\columnwidth]{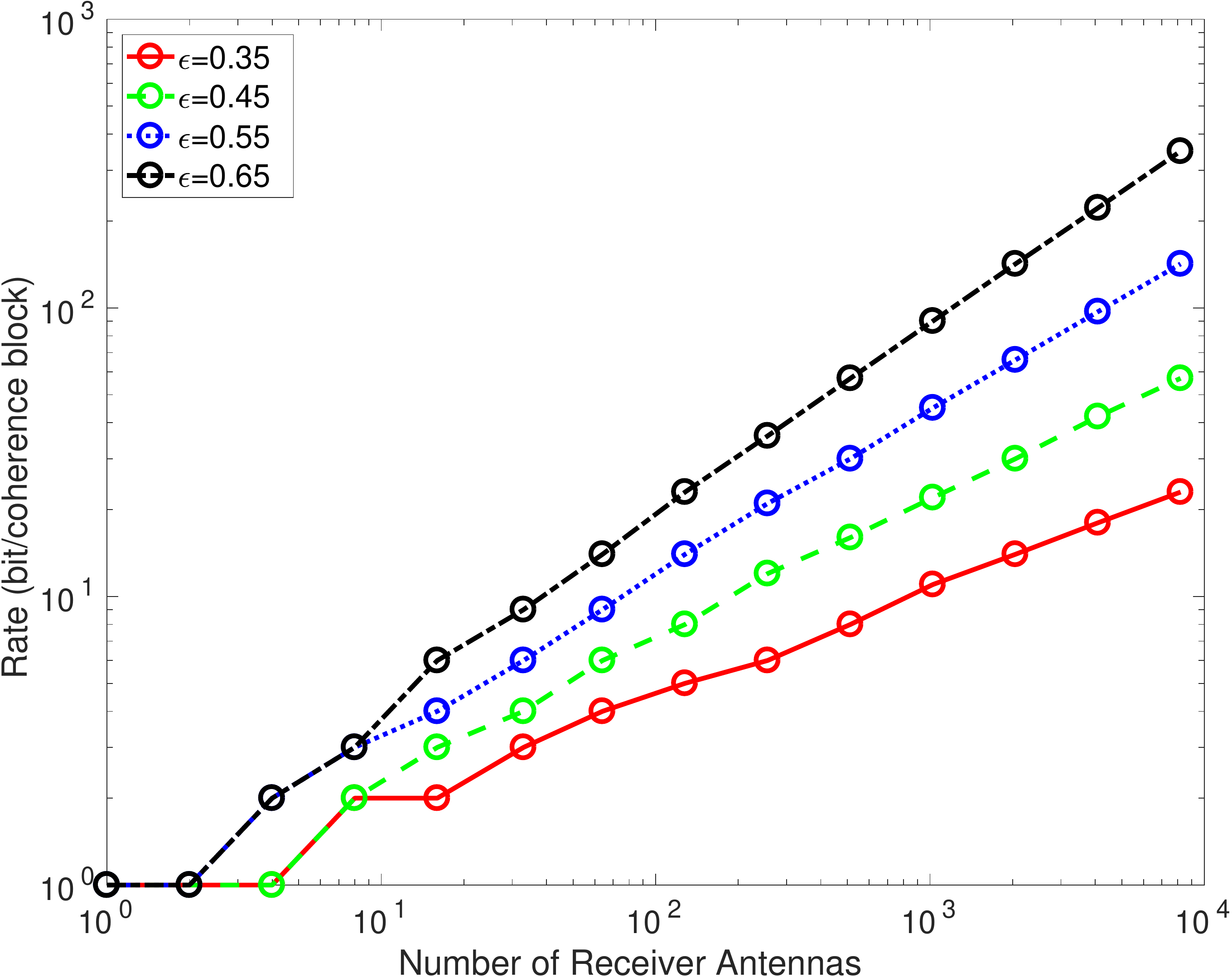}
 \label{fig:EMperformanceRate}
}
\caption{Simulated EM channel for different values of $\epsilon.$ Recall that $\epsilon$ captures the relative growth of available bandwidth (B) with the number of receiver antennas (N), i.e., $B = \Theta((N)^{\epsilon}).$}
 \label{fig:EMperformance}
\end{figure}

\subsection{Numerical plots of achievable rate with bandwidth}
\label{subsec:num_plots}

Figure \ref{fig:genie_vs_pilot} shows a numerical comparison between
the capacities obtained with the genie aided coherent receiver and the
pilot-assisted coherent receiver, as a fixed transmit power
is spread across more and more subchannels.  For the genie-aided
receiver, we assume that the subchannel gains at all antennas are
known precisely and the capacity is \eqref{eq:coherent_cap_erg}.
For the pilot-assisted receiver we perform MMSE channel
estimation using a known pilot symbol sequence with power equal to the average
transmit power per subchannel for 20\% of the available symbol times (i.e., the pilot overhead is 20\%);
and use the remaining time slots for data transmission.  For calculating the rates achievable 
by the pilot-assisted scheme with nearest-neighbor symbol decisions, we use the spectral efficiency analysis for block-fading channels developed in \cite{Jindal2010}, based on an effective SNR \cite[equation (47)]{Jindal2010} that takes into account the channel estimation MMSE \cite[equation (48)]{Jindal2010}. We modify our capacity expression \eqref{eq:coherent_cap_erg} to incorporate this effective SNR rather than the actual SNR, producing $R_{pilot-assisted}=\Ex{}{B B_s \log \left(1+\frac{\frac{P (\sum_{r=0}^{N-1}|H_{r, m}|^2)}{B }(1-MMSE)}{1+\frac{P (\sum_{r=0}^{N-1}|H_{r, m}|^2)}{B }MMSE}\right)}.$
We remark that nearest-neighbor symbol-by-symbol decisions are suboptimal under our channel model
with pilot-assisted channel estimation, but often employed in practical devices \cite{Jindal2010}. 
This rate calculation method illustrates that our scaling law result is applicable not just in asymptotically large regimes of bandwidth and numbers of antennas, but also in practical systems.

\begin{figure}
  \centering
  \includegraphics[width=0.95\linewidth]{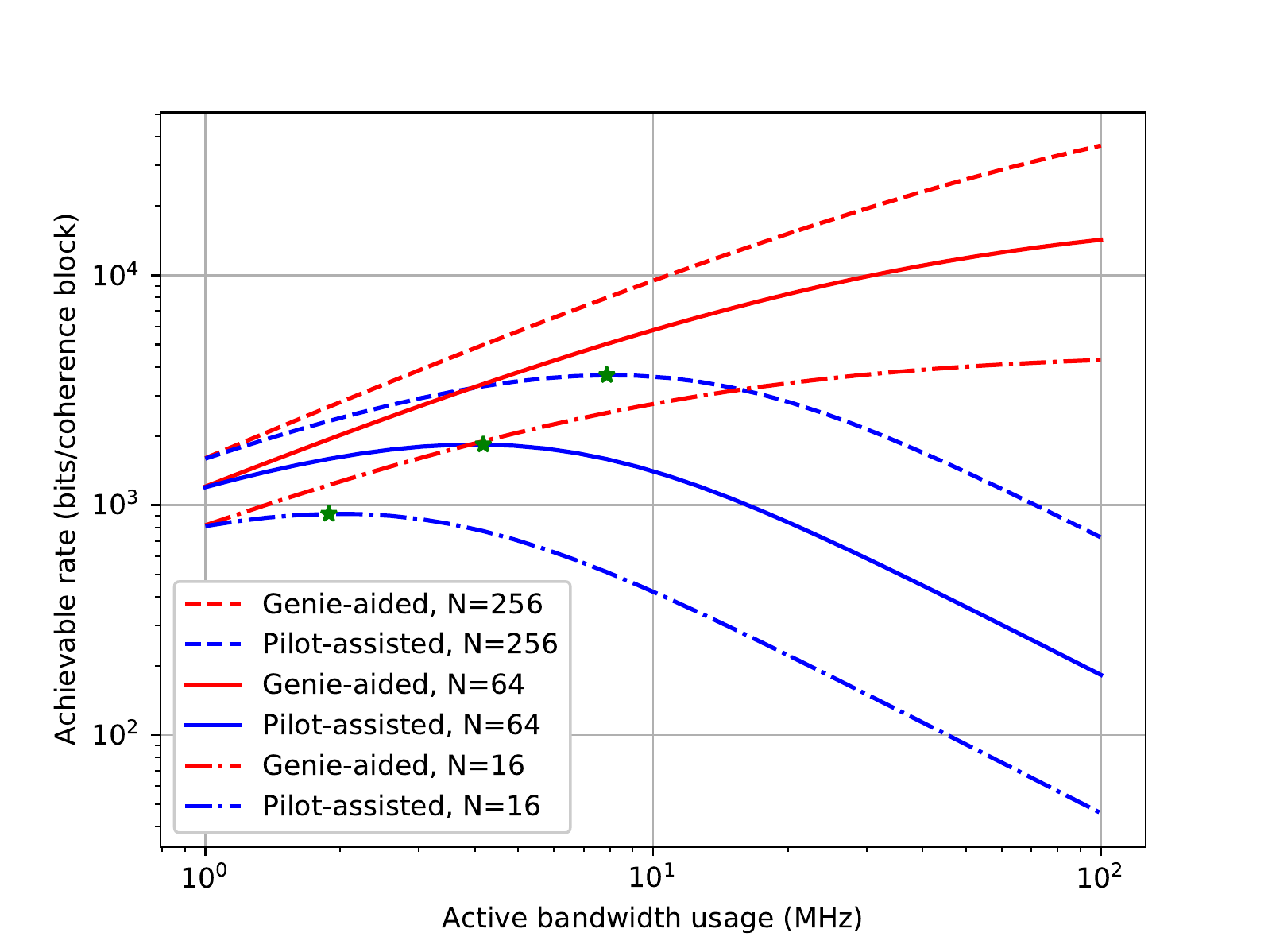}
  \caption{Capacity comparison between genie-aided and pilot-assisted coherent schemes.  $N$ is the number of receiver antennas.}
  \label{fig:genie_vs_pilot}
\end{figure}

We choose a sub-6
GHz band with a total available bandwidth of 100 MHz.  We keep the
coherence time fixed at 200 symbol times, with the subchannel
bandwidth (equal to the coherence bandwidth) being 1 MHz.  This
corresponds to a delay spread of (on the order of) 100 ns, which is
typical for small outdoor deployments.  Generally, both $B$ and $N$
can reach several hundreds in current standards. The 3GPP New Radio
specification has up to $100$ MHz bandwidth \cite{NRrelease15} in the
sub-6 GHz bands (more if we consider carrier aggregation). Arrays with dozens
to hundreds of antennas are practical in sub-6 GHz frequencies use; for example \cite{Bjornson2016}
shows a $3.7$GHz platform with 160 cross-polarized antennas.

The total noise power per subchannel is fixed, and the transmit
power is fixed such that if there is only 1 active subchannel, the
received SNR per subchannel per antenna is 3 dB.  In Fig. \ref{fig:genie_vs_pilot} we consider
a 16-antenna, a 64-antenna and a 256-antenna system.

From the figure, we observe that, while the genie-aided schemes are able to exploit all
available spectrum for capacity improvement, the pilot assisted schemes
are able to exploit only up to a certain bandwidth or, equivalently, a
certain number of subchannels.  If the transmit power is spread over
more than the threshold number of subchannels, the channel estimation
error dominates and drives down the achievable rate.

Furthermore, we represent with green stars in Fig. \ref{fig:genie_vs_pilot}  the threshold bandwidth for the
pilot-assisted receivers, that is, the maximum capacity values of the curves. The thresholds are located at $1.88$, $4.17$ and $7.88$ MHz for 16, 64 and 256 antennas, respectively. Thus for each increase in the number of antennas by a factor of four, we see that
the threshold bandwidth increases roughly by a factor of two.  This is consistent with our analytical derivations
concerning the joint capacity scaling behavior with a large number of antennas and a
large bandwidth, which reveals that for noncoherent receivers (of which
the pilot-assisted coherent receiver is a special case), there is a threshold bandwidth proportional to 
the square root of the number of antennas and any bandwidth
in excess of this threshold does not help in enhancing the capacity scaling.

\subsection{Subchannel bandwidth and coherence bandwidth}
\label{subsec:subchannel_bandwidth}
The gain associated with
the frequency selective wideband channel has been modeled to be
piecewise constant (and i.i.d. Rayleigh fading), with the gain being
constant over the subchannel bandwidth $B_s.$ Hence the subchannel bandwidth $B_s$ is equal
to the coherence bandwidth associated with the channel \cite{goldsmith2005book}. The coherence bandwidth
is a property of the propagation environment, which is usually related to the inverse of
the delay spread.





One possible practical interpretation of our system model is an OFDM system where the subcarrier separation
is exactly $B_s$ and the channel remains constant for $L_c$ OFDM symbols.
In many OFDM systems however, the OFDM subcarrier separation (inverse of
the OFDM symbol time) is much smaller than the coherence bandwidth to
ensure flatness of subcarrier gains \cite{goldsmith2005book}.  On the other hand, our channel model is fully
general and the coherence bandwidth is a property of the block-fading channel whereas the OFDM subcarrier separation
is only a property of one particular waveform. Under a careful interpretation of $B_s$ and $L_c$, the non-coherent capacity bounds in our analysis hold for any waveform that can be converted into our channel model using standard transformations \cite{fgomezUnified}. This includes certain cases of time-domain signaling, OFDM, SC-FDMA, or CDMA.

We defer to future work the analysis of channels where the adjacent subchannel gains are not i.i.d. and multipath channels that do not exhibit rich scattering - such as those in the mmWave bands.

\subsection{Overspreading in more general channel models}
\label{subsec:general_channel_models}

We note that overspreading also appears in channel
models that do not experience rich scattering.  In particular, for a fixed number of antennas,
Telatar and Tse showed that even sparse multipath channels experience
overspreading when the delays of the paths are not known
\cite{journals/tit/TelatarT00}.  Raghavan \textit{et al} found similar
results for capacity scaling with bandwidth under sparse multipath
channels \cite{Raghavan2007}.  Ferrante \textit{et al} revised this
framework to account for mmWave systems with Ricean fading, blockage and
oxygen absorption \cite{Ferrante2016}.

These works suggest that our results might be applicable to a broader
class of channel models than the i.i.d. Rayleigh block-fading model analyzed in this work.

\subsection{Extension to MIMO and Multi-user Channels}
\label{subsec:MUMIMO}

The following remarks briefly discuss how the main observations in this paper for the single-user SIMO channel can be applicable to a MIMO single user channel and to the uplink of cellular systems where multiple single-antenna users transmit to a multi-antenna BS forming a SIMO multiple-access channel (MAC).

\begin{remark}
Consider a massive MIMO channel with the number of transmit antennas $T$ that scales with $N$. Let us denote the transmit antenna index by $t\in\mathcal{T}\triangleq\{1\dots T\}$ and the exponent $T=\Theta(N^\gamma)$. A channel model generalizing \eqref{eq:channelmodel} can be written as follows
\begin{equation}
\label{eq:MIMOchan}
  \begin{split}
    Y_{r, m,\ell}&=\sum_{t=1}^{T}H_{r, m,t} X_{m,\ell,t}+Z_{r, m,\ell}.\\
  \end{split}
\end{equation}
with non-coherent capacity $\displaystyle  C_n(N,B,T)\triangleq \sup_{p(X_{\mathcal{B},\mathcal{L},\mathcal{T}})}\Inf{X_{\mathcal{B},\mathcal{L},\mathcal{T}}}{Y_{\mathcal{N},\mathcal{B},\mathcal{L}}}$. It is shown in \cite{Marzetta1999} that the mutual information does not increase when $T>L_c$, where $L_c$ is the constant coherence block length of the fading channel. Therefore  $\displaystyle C_n(N,B,T)\leq C_n(N,B,\min(L_c,T))$
. Applying the chain rule
we can write the upper bound  $C_n(N,B,T) \leq \min(T,L_c) C_n(N,B)=\Theta(N^{\min(\epsilon,\frac{1}{2})})$, and the lower bound $C_n(N,B,T)\geq C_n(N,B)$. Therefore, when $L_c$ is constant, $T>1$ can increase capacity at most by a constant but the exponent $\gamma$ does not affect the capacity scaling.
\end{remark}

\begin{remark}
 Let us assume a multi-user SIMO channel with $U$ users. From the classic analysis of the MAC channel \cite{cover2006elements} we have that for any set $\mathcal{U}\subset\{1\dots U\}$, the achieved rate of all the users $u\in\mathcal{U}$, denoted as $R_u$, must satisfy the sum-rate constraint $\sum_\mathcal{U} R_u\leq\Inf{X_{\mathcal{B},\mathcal{L},\mathcal{U}}}{Y_{\mathcal{N},\mathcal{B},\mathcal{L}}}$
where $\Inf{X_{\mathcal{B},\mathcal{L},\mathcal{U}}}{Y_{\mathcal{N},\mathcal{B},\mathcal{L}}}$ is the mutual information of a virtual MIMO channel formed by substituting $\mathcal{T}=\mathcal{U}$ in \eqref{eq:MIMOchan}. Therefore, the main observations of our analysis remain applicable to these constraints which, combined, enclose the capacity region of the multi-user SIMO MAC channel.
\end{remark}

\begin{figure*}[!b]
\normalsize
\vspace*{4pt}
\hrulefill
\setcounter{MYtempeqncnt}{\value{equation}}
\setcounter{equation}{26}
\begin{equation}
  \label{eq:1}
  \begin{split}
    \Inf{X_{0,0}}{Y_{\mathcal{N},0,0}}&=\Inf{\sqrt{\overline{A}_{0}}N^{\alpha + \frac{1}{2}}U_{m,\setL}}{Y_{\mathcal{N},0,0}}\\
    &=\Ex{\overline{A}_{0}}{\Ex{ Y_{\setN, 0, 0}|\overline{A}_{0}}{\log \left(\frac{p(Y_{\setN, 0, 0}|\overline{A}_{0} )}{p(Y_{\setN, 0, 0})}\right)}}\\
    &=  \sum_{k \in \mathbb{N}} \Ex{\overline{A}_{0}}{ \Ex{Y_{\setN, 0, 0}|\overline{A}_{0}}{ \log \left( \frac{p(Y_{\setN, 0, 0}|\overline{A}_{0})}{p(Y_{\setN, 0, 0})}\right)}} \\
    &\overset{(a)}{\leq}  \sum_{k \in \mathbb{N}} \Ex{\overline{A}_{0}}{ \mathbf{I}_{\overline{A}_{0} \in \setS_k}\Ex{Y_{\setN, 0, 0}|\overline{A}_{0}}{ \log \left( \frac{p(Y_{\setN, 0, 0}|\overline{A}_{0})}{\int_{s \in \setS_k}  p(Y_{\setN, 0, 0}|\overline{A}_0 = s)\mu_{\overline{A}_0}( ds)} \right)}}\\
    &=-\sum_{k \in \mathbb{N}} \Ex{\overline{A}_{0}}{ \mathbf{I}_{\overline{A}_{0} \in \setS_k}\Ex{Y_{\setN, 0, 0}|\overline{A}_{0}} {\log \left(\int_{s\in I_k} \frac{p(Y_{\setN, 0, 0}|\overline{A}_0 = s)}{p(Y_{\setN, 0, 0}|\overline{A}_{0})} \mu_{\overline{A}_0}(ds)\right) }}.
  \end{split}
\end{equation}
\setcounter{equation}{\value{MYtempeqncnt}}
\end{figure*}

\section{Conclusions and Future Work}
\label{sec:conclusion}

We have derived the capacity scaling laws governing the relation
between large bandwidth and number of receiver antennas in coherent
and non-coherent i.i.d. Rayleigh block-fading SIMO channels. We have
found that in a block-fading channel of a fixed coherence length, the
non-coherent capacity scales with the minimum of the number of
independent subchannels and the square root of the number of receive
antennas.  On the other hand, coherent capacity under perfect
genie-aided channel state information scales linearly with the
bandwidth.  We have also shown that the capacity scaling of the
non-coherent channel with a large bandwidth and large number of
antennas can be achieved by an energy modulation scheme with a
non-coherent receiver.

Our results shed light on the inherent difficulty of supporting uplink
wideband communications using a large antenna array at the receiver
under a transmit power constraint.  Our capacity scaling law
characterizations reveal that if the number of antennas is large, and
the channel realizations are not known apriori, using excessive bandwidth can
actually degrade the achievable rates.  This is fundamentally due to
the fact that with a larger bandwidth, and hence, a larger number of
subchannels, there is less power available for pilot-assisted channel
estimation at each subchannel.  This is in sharp contrast to what
might be expected from a coherent capacity analysis with genie-aided
perfect channel state information.

We finally point out that our results have focused on the scaling laws
of capacities in the joint asymptotic regime of large bandwidth and large
number of antennas. The capacity and capacity-achieving strategy of
non-coherent systems with a finite number of antennas and bandwidth is
unknown, and hence this work does not answer questions about the capacity
difference between coherent and noncoherent systems in this
regime. Characterizing various properties of the non-coherent channel
capacity and the capacity achieving transmission strategy remains an
interesting topic for future work.


\appendices

\section{Mutual information upper bound for  $L_c = 1$}
\label{app:lemm:separation}

In this appendix we prove Lemma \ref{lemm:opt_mi}.
Lemma \ref{lem:marzetta} tells us that the optimal input satisfies $X_{m,\setL}=\sqrt{A_m}U_{m,\setL}$ where $A_m\geq0$ and $U_{m,\setL}$ is a unitary isotropic vector of length $L_c$. For convenience we define a normalized amplitude distribution as $\overline{A}_{m} \triangleq (\sum_{\ell \in \setL}|X_{m, \ell}|^2) N^{\alpha + \frac{1}{2}}=A_mN^{\alpha + \frac{1}{2}},$ so that we have $\Ex{}{\overline{A}_{0}}= \Theta(1)$.

Since $L_c=1$ the set $\setL=\{0\dots L_c-1\}$ becomes just the first element, $\{0\}$. Therefore, we replace the subindex $\setL$ by $0$ in all symbols hereafter. For example, $U_{m,0}$ is a unitary isotropic vector of length $1$, which boils down to a uniform random phase rotation. $\overline{A}_{m}$ simplifies as $|X_{m, 0}|^2 N^{\alpha + \frac{1}{2}}=A_mN^{\alpha + \frac{1}{2}},$ and so on.

We want to upper bound the subchannel capacity function \eqref{eq:MIfun} with $\setL=\{0\}$ and the power allocation argument $\rho_m=\frac{P}{N^{\frac{1}{2}+\alpha}}$, which produces

  \begin{equation}f_{\mathrm{C}}(\frac{P}{N^{\frac{1}{2}+\alpha}})\triangleq\max_{p(X_{0,0}):\Ex{}{|X_{0,0}|^2} \leq  \frac{P}{N^{\frac{1}{2}+\alpha}}}\Inf{X_{0,0}}{Y_{\mathcal{N},0,0}}.
  \end{equation}
Here without loss of generality we have chosen the subchannel index $m=0$ to evaluate the mutual information, which takes the same value in all the ``active'' $M=N^{\frac{1}{2}+\alpha}\leq B$ subchannels.

For any $\xi <\alpha,$ we define the sets $\setS_k = [(k-1)N^{\alpha - \xi}, k N^{\alpha - \xi})$ for $k$ belonging to the set of natural numbers $k\in\mathbb{N}$.  We can verify that $ \cup_{k \in \mathbb{N}} \setS_k$ covers the positive real line. We upper bound the mutual information using the sets $\setS_k$ and $\overline{A}_{0}$ and the ``indicator function'' defined as

  \begin{equation}\mathbf{I}_{\overline{A}_0 \in \setS_k}=\begin{cases}
                                     1& a_0 \in \setS_k\\
                                     0& a_0 \notin \setS_k,\\
                                    \end{cases}
  \end{equation}
which leads to the upper bound on the mutual information integral given in \eqref{eq:1}.

\stepcounter{equation}


The inequality $(a)$ follows from the fact that $p(Y_{\setN, 0, 0}) \leq \int_{s \in \setS_k} p(Y_{\setN, 0, 0}|\overline{A}_0 = s) \mu_{\overline{A}_0}(ds),$ where $\mu_{X}(\mathcal{A})$ refers to the measure of set $\mathcal{A}$ under the distribution induced by the random variable $X.$ And the last step uses the negative of the logarithm so that we can introduce the term $p(Y_{\setN, 0, 0}|\overline{A}_{0})$ inside the integral over $ds$.

Say we denote the integration variable for the average $\Ex{\overline{A}_0}{.}$ in the above expression with the letter $a$. This step has allowed us to write the upper bound of the mutual information as an integral containing $\frac{p(Y_{\setN,0,0} |\overline{A}_0=s)}{p(Y_{\setN,0,0} |\overline{A}_0=a)}$ where $a$ and $s$ are treated as independent realizations of the random variable $\overline{A}_0$ at two different integrals.

\begin{figure*}[!b]
\normalsize
\vspace*{4pt}
\hrulefill
\setcounter{MYtempeqncnt}{\value{equation}}
\setcounter{equation}{31}
\begin{equation}
  \label{eq:7}
  \begin{split}
    \Inf{X_{0,0}}{Y_{\mathcal{N},0,0}}&=  \Ex{\overline{A}_{0}}{\Ex{ Y_{\setN, 0, 0}|\overline{A}_{0}}{\log \left(\frac{p(Y_{\setN, 0, 0}|\overline{A}_{0} )}{p(Y_{\setN, 0, 0})}\right)}} \\
    &\overset{(a)}{\leq}  -\sum_{k=1}^{\infty} \int_{a \in \setS_k}
    \Ex{Y_{\setN, 0, 0}|\overline{A}_0 = a}{ \log \int_{s\in \setS_k} \frac{1}{p_k}
      \frac{p(Y_{\setN, 0, 0}|\overline{A}_0 = s)}{p(Y_{\setN, 0, 0}|\overline{A}_0=a)} \mu_{\overline{A}_0}(ds)} \mu_{\overline{A}_0}(da)\\
    &\overset{(b)}{=}\sum_{k=1}^{\infty} - \int_{a \in \setS_k} \Ex{Y_{\setN, 0, 0}|a}{ \log \left( 1 -
        \Theta\left(\int_{s \in \setS_k} \frac{1}{p_k} T(s, G, a)
          \mu_{\overline{A}_0}(ds)\right)\right)} \mu_{\overline{A}_0}(da)\\
    &\overset{(c)}{\leq}\sum_{k=1}^{\infty} \int_{a\in \setS_k} \int_{s \in \setS_k} \Theta(N)
      \left(\overbw{s -a}\right) \left( 1 - \Ex{Y_{\setN, 0, 0}|a}{G} \right) \mu_{\overline{A}_0}(ds) \mu_{\overline{A}_0}(da)\\
    &\overset{(d)}{=}\sum_{k=1}^{\infty} \int_{a\in \setS_k} \int_{s \in \setS_k} N
      \left(\overbw{s -a}\right) \left( 1 - 1 - \overbw{a} \right) \mu_{\overline{A}_0}(ds) \mu_{\overline{A}_0}(da)\\
    &\overset{(e)}{=}\sum_{k=1}^{\infty} \int_{a\in \setS_k} \int_{s \in \setS_k}N^{-2\alpha}
      \left((a -s) a\right)  \mu_{\overline{A}_0}(ds) \mu_{\overline{A}_0}(da),\\
  \end{split}
\end{equation}
\setcounter{equation}{\value{MYtempeqncnt}}
\end{figure*}

We now show that in the limit $N\to\infty$ the conditional distributions $p(Y_{\setN, 0, 0}|\overline{A}_{0}=s)$ and $p(Y_{\setN, 0, 0}|\overline{A}_{0}=a)$ are similar for all $a,s \in \setS_k$.  In particular we expand the ratio as follows

\begin{equation}
  \label{eq:2}
  \begin{split}
    \frac{p(Y_{\setN, 0, 0}|\overline{A}_0 = s)}{p(Y_{\setN, 0, 0}|\overline{A}_0 = a)} &= \left( \frac{
        1+\overbw{a} }{
        1+\overbw{s}} \right)^N\times\\
        &e^{- (\sum_{r=0}^{N-1}|Y_{r, 0, 0}|^2 ) \left(\frac{1}{1+ \overbw{s}} - \frac{1}{1+\overbw{a}}\right)}.
  \end{split}
\end{equation}
Defining the difference $\Delta a = s-a,$ this may be written as

\begin{equation}
  \label{eq:3}
  \begin{split}
    \frac{p(Y_{\setN, 0, 0}|\overline{A}_0 = s)}{p(Y_{\setN, 0, 0}|\overline{A}_0 = a)} &= \left( 1 + \frac{\overbw{\Delta
        a}}{1+ \overbw{a}}
    \right)^{-N}
    \times\\
        &e^{-N G\left(\frac{-\overbw{\Delta a}}{\left(1+ \overbw{a})(1+ \overbw{a} + \overbw{\Delta a}\right)} \right)}
    \\    &
    \doteq e^{-N \overbw{\Delta a}(1-G) }.
  \end{split}
\end{equation}
where $G \triangleq \frac{\sum_{r=0}^{N-1} |Y_{r,0,0}|^2}{N},$ and we
use the result that $(1+f(N))^{t(N)} \doteq e^{f(N)t(N)}$ if $f(N)
\rightarrow 0, t(N) \rightarrow \infty,$ such that $f(N) t(N)
\rightarrow \infty.$ We now observe that, as $N \rightarrow \infty$,
by the central limit theorem, $\sqrt{N}(G - 1 - \overbw{a})$ converges
to a zero mean normal random variable independent of $a$ with a
variance independent of $N$.

We can now use \eqref{eq:3} to write the part inside the logarithm as follows

\begin{equation}
  \label{eq:5}
  \begin{split}
    &\int_{s\in \setS_k} \frac{1}{p_k}
    \frac{p(Y_{\setN, 0, 0}|\overline{A}_0 = s)}{p(Y_{\setN, 0, 0}|\overline{A}_{0} = a)} \mu_{\overline{A}_0}(ds)\\
    &\quad = \int_{s \in
      I_k} \frac{1}{p_k} e^{-T(s, G, \overline{A}_{0})} \mu_{\overline{A}_0}(ds) \\
    &\quad = 1 - \Theta\left(\int_{s \in \setS_k} \frac{p(s)}{p_k} T(s, G, a)
    ds\right),
  \end{split}
\end{equation}
where $ T(s, G, a) = \left( N \overbw{(s-a)} (1-G) \right).$ The last
step follows from the CLT and the observation that

  \begin{equation}\lim_{N\rightarrow \infty}
T(s, G, a) = \lim_{N \rightarrow \infty} N \overbw{s-a} \overbw{a} =  0.
  \end{equation}

Writing down all the steps in the integral $\Ex{\overline{A}_0}{.}$ with variable $a \in \setS_k$ we get \eqref{eq:7}, 
where $(a)$ is \eqref{eq:1} and $(b)$ comes from \eqref{eq:2},\eqref{eq:3},\eqref{eq:5}. In $(c)$ we use $\log(1-x) = -\Theta(x)$ as $x\rightarrow 0$. In $(d)$ we use the average of $G$ according to CLT, and $(e)$ is just cleanup.
\stepcounter{equation}

Using $\int \mu_{\overline{A}_0}(ds) = 1,$ $(\int a  \mu_{\overline{A}_0}(da))(\int s \mu_{\overline{A}_0}(ds)) = \Theta(1),$ and $|s-a| < \Theta(N^{\alpha - \xi}),$ due to the assumption that $a,s
\in \setS_k,$ we get that,

  \begin{equation}    \Ex{\overline{A}_{0}}{\Ex{ Y_{\setN, 0, 0}|\overline{A}_{0}}{\log \left(\frac{p(Y_{\setN, 0, 0}|\overline{A}_{0} )}{p(Y_{\setN, 0, 0})}\right)}} \leq \Theta(N^{-\alpha - \xi}).
  \end{equation}

\begin{figure*}[!b]
\normalsize
\vspace*{4pt}
\hrulefill
\setcounter{MYtempeqncnt}{\value{equation}}
\setcounter{equation}{38}

\begin{align}
\begin{split}
  \eqref{eq:39}&\doteq \left(1 + \overbw{b - a}\right)^N e^{ \sum_{r=0}^{N-1}
      ((|\sum_{\ell=0}^{L_c-1} Y_{r, 0, \ell} u_{0,
          \ell}^*|^2)(\overbw{a}(1- \overbw{a}))  - (|\sum_{\ell=0}^{L_c-1} Y_{r, 0, \ell} u_{0,
          \ell}^*|^2)(\overbw{b}(1 - \overbw{b}) ))}\\
  &\doteq e^{N\overbw{b - a}} e^{ \sum_{r=0}^{N-1}
      ((|\sum_{\ell=0}^{L_c-1} Y_{r, 0, \ell} u_{0,
          \ell}^*|^2)\overbw{a}  - (|\sum_{\ell=0}^{L_c-1} Y_{r, 0, \ell} v_{0,
          \ell}^*|^2)\overbw{b})}.\\
\end{split}
\label{eq:intermediate}
\end{align}
\setcounter{equation}{\value{MYtempeqncnt}}
\end{figure*}
  
Since this must be satisfied for any $\xi\leq \alpha$ the tightest bound is choosing $\xi=\alpha$ producing

  \begin{equation}
  \begin{split}
   \Inf{X_{0,0}}{Y_{\mathcal{N},0,0}}&= \Ex{\overline{A}_{0}}{\Ex{ Y_{\setN, 0, 0}|\overline{A}_{0}}{\log \left(\frac{p(Y_{\setN, 0, 0}|\overline{A}_{0} )}{p(Y_{\setN, 0, 0})}\right)}}\\
   &\leq \Theta(N^{-2\alpha}).
  \end{split}
  \end{equation}

\section{Mutual information upper bound for $L_c \geq 1$}
\label{app:lemm:separation_block}

Lemma \ref{lem:marzetta} specifies that the optimal input satisfies $X_{m,\setL}=\sqrt{A_m}U_{m,\setL}$ where $A_m>0$ and $U_{m,\setL}$ is a unitary isotropic vector of length $L_c$. It is known that

  \begin{equation}\sum_{\ell\in\setL}|U_{m,\ell}|^2=1.
  \end{equation}

We next write down the distribution of the output for subchannel $m=0$, denoted $Y_{\setN, 0, \setL}$, given a fixed input $X_{0, \setL}.$ We have

\begin{align}
\begin{split}
&p(Y_{\setN, 0, \setL}|X_{0, \setL}) = \\
&\quad= \frac{ \int e^{-\sum_{r=0}^{N-1} \sum_{\ell=0}^{L_c - 1}|Y_{r, 0, \ell} - H_{r,0} X_{0, \ell}|^2 } e^{-\sum_{r=0}^{N-1}|H_{r, 0}|^2} dH_{\setN}}{\pi^{N L_c}}\\
 &\quad= \frac{ \int e^{-\sum_{r=0}^{N-1} \sum_{\ell=0}^{L_c - 1}|Y_{r, 0, \ell} - H_r \sqrt{A_0} U_{0, \ell}|^2 } e^{-\sum_{r=0}^{N-1}|H_r|^2} dH_{\setN}}{\pi^{N L_c}}\\
 &\quad= \frac{e^{(-\sum_{r=0}^{N-1} \sum_{\ell=0}^{L_c - 1}|Y_{r, 0, \ell}|^2) + \sum_{r=0}^{N-1} \frac{|\sum_{\ell=0}^{L_c-1} Y_{\ell} \sqrt{A_0} U_{0, \ell}^*|^2}{1 + A_0} }}{\pi^{N  L_c}(1+A_0)^N}  \\
\end{split}
\label{eq:block_conditional}
\end{align}

Similar to the proof for $L_c = 1,$ we now proceed to define a normalized amplitude distribution as $\overline{A}_{0} \triangleq (\sum_{\ell \in \setL}|X_{0, \ell}|^2) N^{\alpha + \frac{1}{2}}=A_0N^{\alpha + \frac{1}{2}}.$ Observe that

  \begin{equation} \Ex{}{\overline{A}_0} = L_c = \Theta(1).
  \end{equation}
We now proceed to find the ratio of two conditional
distributions for two different realizations of $X_{0, \setL}.$ As in the case $L_c=1$, such a ratio will be found inside an average where $a$ and $b$ (respectively $u_{0,\setL}$ and $v_{0,\setL}$) will be integration variables treated as independent realizations of the variable $\overline{A}_0$ (respectively $\U_{0,\setL}$) for integration purposes. We have:

\begin{align}
\label{eq:39}
\begin{split}
  &\frac{p(Y_{\setN, 0, \setL}|\overline{A}_0 = a, U_{0, \setL} =
    u_{0, \setL})}{p(Y_{\setN, 0, \setL}|\overline{A}_0 = b, U_{0, \setL} =
    v_{0, \setL})} =\\
  &\quad= \frac{(1 + \overbw{b})^N}{(1+\overbw{a})^N}
  \frac{e^{ \sum_{r=0}^{N-1} \frac{|\sum_{\ell=0}^{L_c-1} Y_{r, 0, \ell} \sqrt{\overbw{a}}
        u_{0, \ell}^*|^2}{1 + \overbw{a}} }}{e^{\sum_{r=0}^{N-1}
      \frac{|\sum_{\ell=0}^{L_c-1} Y_{r, 0, \ell} \sqrt{\overbw{b}} v_{0, \ell}^*|^2}{1 +
        \overbw{b}}}}.
\end{split}
\end{align}

We now make an observation about the exponents.  We note that, if $a,b
\in \setS_k$ for a fixed $k,$ $\overbw{a}, \overbw{b} \rightarrow 0$
as $N \rightarrow \infty, $ and we thus obtain \eqref{eq:intermediate} where the notation $f(N)\dot{=}g(N)$ represents $\lim_{N\rightarrow\infty}f(N)-g(N)=0$.
\stepcounter{equation}

We now observe the following: under the distribution $p(Y_{\setN, 0, \setL}|\overline{A}_0 = b, U_{0, \setL} =
    v_{0, \setL})$, we have

  \begin{equation}|\sum_{\ell=0}^{L_c-1} Y_{r, 0, \ell} u_{0,
          \ell}^*| = |\sqrt{b \overbw{|\sum_{\ell \in \setL}u_{0, \ell}^* v_{0, \ell}|^2}} + \tilde{Z_0}|,
  \end{equation} and

  \begin{equation}\sum_{\ell=0}^{L_c-1} Y_{r, 0, \ell} v_{0,
   \ell}^* = \sqrt{\overbw{b}} + \hat{Z_0},
  \end{equation} where $\tilde{Z_0}$ and
 $\hat{Z_0}$ are zero mean complex Gaussians of unit second moment, with $\Ex{}{\tilde{Z_0} \hat{Z_0}^*} = \sum_{\ell=0}^{L_c-1} u_{0, \ell} v_{0,\ell}^* \triangleq t(u_{0,\setL}, v_{0,\ell}).$
 By concentration of measure, we have that

\begin{equation}
 \begin{split}
  G &\triangleq \begin{array}{clc}\frac{1}{N} \sum_{r=0}^{N-1} \Big(&(|\sum_{\ell=0}^{L_c-1}
  Y_{r, 0, \ell} u_{0, \ell}^*|^2)\overbw{a} &\\&\qquad - (|\sum_{\ell=0}^{L_c-1}
  Y_{r, 0, \ell} v_{0,
    \ell}^*|^2)\overbw{b}&\Big)\end{array}\\
  &\sim \overbw{a - b} (1 + \frac{m(u_{0,\setL}, v_{0,\ell})}{ \sqrt{N}}\nu)
  \\ &\qquad+ \overbw{a} \overbw{b} |\sum_{\ell \in \setL}u_{0, \ell}^* v_{0, \ell}|^2 
  \\ &\qquad- (\overbw{b})^2+ \text{lower order terms,}
 \end{split}
\end{equation}

where $\sim$ stands for ``distributed as'', $\nu$ (defined as
$\frac{\sqrt{N}(G - 1)}{m(u_{0,\setL}, v_{0,\setL})}$) is distributed
as a Gaussian random variable independent of $u_{0, \setL}$ and $v_{0,
  \setL},$ and

  \begin{equation} m^2(u_{0,\setL}, v_{0,\setL}) = (a + b)^2 - 2 a b
\Ex{}{|\tilde{Z_0}\hat{Z_0}|^2} .
  \end{equation} Observe that, when $t(u_{0,\setL},
v_{0,\ell}) = 1, $ i.e., when $u_{0, \setL} = v_{0, \setL},$ we have
$m^2(u_{0,\setL}, v_{0,\ell}) = (a-b)^2,$ and when $t(u_{0, \setL},
v_{0, \setL}) = 0,$ $m^2(u_{0,\setL}, v_{0,\setL}) = a^2 + b^2.$
Thus

  \begin{equation}(a-b)^2 \leq m^2(u_{0, \setL}, v_{0, \setL}) \leq a^2 + b^2,
  \end{equation}
and the maximum variation in the value of $m^2(u_{0, \setL}, v_{0,
  \setL})$ across all values of $u_{0, \setL}$ is bounded by
$\Theta(ab).$

\begin{figure*}[!b]
\normalsize
\vspace*{4pt}
\hrulefill
\setcounter{MYtempeqncnt}{\value{equation}}
\setcounter{equation}{45}

\begin{align}
\begin{split}
  &\frac{\log(\int \frac{p(Y_{\setN, 0, \setL}|\overline{A}_0 = a, U_{0, \setL} = u_{0,
      \setL})}{p(Y_{\setN, 0, \setL}|\overline{A}_0 = b, U_{0, \setL} =
    v_{0, \setL})} \mu_{U_{0, \setL}}(du_{0, \setL}))}{ \max_{u_{0, \setL}} N\left(\overbw{a-b} \frac{m(u_{0,\setL}, v_{0,\setL})\nu}{\sqrt{N}} - \overbw{b}\left(\overbw{b}
    - \overbw{a}|\sum_{\ell \in \setL}u_{0, \ell}^* v_{0, \ell}|^2\right)\right)} \rightarrow 1.\\
\end{split}
\label{eq:intermediate_mod}
\end{align}

\vspace*{4pt}
\hrulefill

\setcounter{equation}{47}

\begin{align}
\label{eq:48}
\begin{split}
  \Inf{X_{0,0}}{Y_{\mathcal{N},0,\setL}}&=\Ex{\overline{A}_0, U_{0, \setL}}{\Ex{Y_{\setN, 0, \setL}|\overline{A}_0, U_{0, \setL}}{\log\left(\frac{p(Y_{\setN, 0, \setL}|\overline{A}_0, U_{0, \setL})}{p(Y_{\setN, 0, \setL})}\right)}}\\
  &\leq \sum_{k = 1}^{\infty}\Ex{\overline{A}_0, U_{0, \setL}}{\mathbf{I}_{\overline{A}_0 \in \setS_k}\Ex{Y_{\setN, 0, \setL}|\overline{A}_0, U_{0, \setL}}{\log\left(\frac{p(Y_{\setN, 0, \setL}|\overline{A}_0, U_{0, \setL})}{ \int_{a\in \setS_k}  \int_u p(Y_{\setN, 0, \setL}|a, u) \mu_{U_{0,\setL}}( du) \mu_{\overline{A}_0}( da)}\right)}}\\
  &\overset{(a)}{=} \sum_{k = 1}^{\infty}\Ex{\overline{A}_0, U_{0, \setL}}{\mathbf{I}_{\overline{A}_0 \in \setS_k}\Ex{Y_{\setN, 0, \setL}|\overline{A}_0, U_{0, \setL}}{-\int_{a\in \setS_k} M(a, \overline{A}_0, \nu)\mu_{\overline{A}_0}(da)}}\\
  &\overset{(b)}{=} \sum_{k\in \mathbb{N}} \int_{b \in \setS_k} \int_{a\in \setS_k} \left(\frac{b^2}{N^{ 2\alpha}} - \frac{1}{2} \frac{ab}{N^{2\alpha}}\right)   \mu_{\overline{A}_0}(da) \mu_{\overline{A}_0}(db)\\
&\overset{(c)}{=} \Theta(N^{-2 \alpha}).
\end{split}
\end{align}

\setcounter{equation}{\value{MYtempeqncnt}}
\end{figure*}

We can then simplify \eqref{eq:intermediate} to give

\begin{equation}
\label{eq:intermediate2}
\begin{split}
  &\frac{p(Y_{\setN, 0, \setL}|\overline{A}_0 = a, U_{0, \setL}=
    u_{0, \setL})}{p(Y_{\setN, 0, \setL}|\overline{A}_0 = b, U_{0, \setL} =
    v_{0, \setL})}
   \\
&\quad \doteq e^{\overbw{(b - a)N + \sum_{r=0}^{N-1}
      \left(a(|\sum_{\ell=0}^{L_c-1} Y_{r, 0, \ell} u_{0,
          \ell}^*|^2)  - b(|\sum_{\ell=0}^{L_c-1} Y_{r, 0, \ell} v_{0,
          \ell}^*|^2)\right)}}\\
&\quad\doteq e^{N \left(\overbw{b-a}+G\right)}\\
&\quad\doteq e^{N\left(\overbw{a-b} \frac{m(u_{0,\setL}, v_{0,\ell})\nu}{\sqrt{N}} - \overbw{b}\left(\overbw{b-a|\sum_{\ell \in \setL}u_{0, \ell}^* v_{0, \ell}|^2}\right)\right)}.
\end{split}
\end{equation}
From Laplace's principle, we then get \eqref{eq:intermediate_mod}
as $N \rightarrow \infty.$
\stepcounter{equation}
Observe that the denominator is independent
of $v_{0, \setL}$ in the limit as $N \rightarrow \infty.$   We now observe that

\begin{align}
  \label{eq:exponent_simpl}
  \begin{split}
    & \begin{array}{clc}\max_{u_{0, \setL}}N\Big(&\overbw{a-b} \frac{m(u_{0,\setL}, v_{0,\setL})\nu}{\sqrt{N}} &\\&\qquad - \overbw{b}\left(\overbw{b
    - a|\sum_{\ell \in \setL}u_{0, \ell}^* v_{0, \ell}|^2}\right)&\Big)\end{array} \\
&\triangleq M(a, b, \nu)\\
    & =
    \begin{cases}
     N\left( \overbw{a-b} \frac{|a-b| \nu}{\sqrt{N}}  - \overbw{b}\left(\overbw{b-a}\right)\right) &\text{ if } (a-b)\nu < 0, \\
      N\left(\overbw{a-b} +\frac{\sqrt{a^2 + b^2} \nu}{\sqrt{N}}  - \left(\overbw{b}\right)^2\right) &\text{ if } (a-b)\nu > 0,
    \end{cases}
  \end{split}
\end{align}
for a large enough $N.$

We now bound the mutual information for subchannel
$0$. Using the same definition of $\setS_k$ that we had in Appendix
\ref{app:lemm:separation}, we have that the mutual information in the
subchannel 0 is given by \eqref{eq:48}
\stepcounter{equation}
where $(a)$ follows from \eqref{eq:exponent_simpl} and the observation
that $\nu$ in \eqref{eq:exponent_simpl} is independent of $U_{0,
  \setL}$ in the limit.  $(b)$ follows from the fact that
$\Ex{\nu|b}{M(a, b, \nu)} = - \frac{b^2}{N^{ 2\alpha}} + \frac{1}{2}
\frac{ab}{N^{ 2\alpha}} $.  $(c)$ is established exactly the same way
as in Appendix \ref{app:lemm:separation}. This concludes the proof.


\begin{thebibliography}{10}
\providecommand{\url}[1]{#1}
\csname url@samestyle\endcsname
\providecommand{\newblock}{\relax}
\providecommand{\bibinfo}[2]{#2}
\providecommand{\BIBentrySTDinterwordspacing}{\spaceskip=0pt\relax}
\providecommand{\BIBentryALTinterwordstretchfactor}{4}
\providecommand{\BIBentryALTinterwordspacing}{\spaceskip=\fontdimen2\font plus
\BIBentryALTinterwordstretchfactor\fontdimen3\font minus
  \fontdimen4\font\relax}
\providecommand{\BIBforeignlanguage}[2]{{%
\expandafter\ifx\csname l@#1\endcsname\relax
\typeout{** WARNING: IEEEtran.bst: No hyphenation pattern has been}%
\typeout{** loaded for the language `#1'. Using the pattern for}%
\typeout{** the default language instead.}%
\else
\language=\csname l@#1\endcsname
\fi
#2}}
\providecommand{\BIBdecl}{\relax}
\BIBdecl

\bibitem{mainak2015wideband}
M.~Chowdhury, A.~Manolakos, F.~G{\'{o}}mez-Cuba, E.~Erkip, and A.~J. Goldsmith,
  ``{Capacity scaling in noncoherent wideband massive SIMO systems},'' in
  \emph{IEEE Information Theory Workshop (ITW)}, 2015.

\bibitem{goldsmith2005book}
A.~Goldsmith, \emph{{Wireless communications}}.\hskip 1em plus 0.5em minus
  0.4em\relax Cambridge University Press, 2005.

\bibitem{Marzetta1999}
T.~L. Marzetta and B.~M. Hochwald, ``{Capacity of a mobile multiple-antenna
  communication link in Rayleigh flat fading},'' \emph{IEEE Transactions on
  Information Theory}, vol.~45, no.~1, pp. 139--157, 1999.

\bibitem{1337103}
M.~Katz and S.~Shamai, ``{On the capacity-achieving distribution of the
  discrete-time noncoherent and partially coherent AWGN channels},'' \emph{IEEE
  Transactions on Information Theory}, vol.~50, no.~10, pp. 2257--2270, 2004.

\bibitem{Perera2006}
R.~R. Perera, T.~S. Pollock, and T.~D. Abhayapala, ``{Non-coherent Rayleigh
  fading MIMO channels: Capacity and optimal input},'' in \emph{IEEE
  International Conference on Communications (ICC)}, 2006, pp. 4180--4185.

\bibitem{Lozano2008}
A.~Lozano, ``{Interplay of spectral efficiency, power and Doppler spectrum for
  reference-signal-assisted wireless communication},'' \emph{IEEE Transactions
  on Wireless Communications}, vol.~7, no.~12, pp. 5020--5029, 2008.

\bibitem{Sethuraman2009}
V.~Sethuraman, L.~Wang, B.~Hajek, and A.~Lapidoth, ``{Low-SNR capacity of
  noncoherent fading channels},'' \emph{IEEE Transactions on Information
  Theory}, vol.~55, no.~4, pp. 1555--1574, 2009.

\bibitem{Jindal2010}
N.~Jindal and A.~Lozano, ``{A unified treatment of optimum pilot overhead in
  multipath fading channels},'' \emph{IEEE Transactions on Communications},
  vol.~58, no.~10, pp. 2939--2948, 2010.

\bibitem{7541626}
M.~Chowdhury and A.~Goldsmith, ``{Capacity of block Rayleigh fading channels
  without CSI},'' in \emph{2016 IEEE International Symposium on Information
  Theory (ISIT)}, 2016, pp. 1884--1888.

\bibitem{journals/tit/TelatarT00}
{\`{I}}.~E. Telatar and D.~N. Tse, ``{Capacity and mutual information of
  wideband multipath fading channels},'' \emph{IEEE Transactions on Information
  Theory}, vol.~46, no.~4, pp. 1384--1400, 2000.

\bibitem{journals/tit/MedardG02}
M.~M{\'{e}}dard and R.~G. Gallager, ``{Bandwidth scaling for fading multipath
  channels},'' \emph{IEEE Transactions on Information Theory}, vol.~48, no.~4,
  pp. 840--852, 2002.

\bibitem{journals/tit/Verdu02}
S.~Verd{\'{u}}, ``{Spectral efficiency in the wideband regime},'' \emph{IEEE
  Transactions on Information Theory}, vol.~48, no.~6, pp. 1319--1343, 2002.

\bibitem{Medard2005}
M.~Medard, ``{On approaching wideband capacity using multitone FSK},''
  \emph{IEEE Journal on Selected Areas in Communications}, vol.~23, no.~9, pp.
  1830--1838, 2005.

\bibitem{Zheng2007noncoherent}
L.~Zheng, D.~N.~C. Tse, and M.~Medard, ``{Channel coherence in the low-SNR
  regime},'' \emph{IEEE Transactions on Information Theory}, vol.~53, no.~3,
  pp. 976--997, 2007.

\bibitem{Ray2007noncoherent}
S.~Ray, M.~Medard, and L.~Zheng, ``{On noncoherent MIMO channels in the
  wideband regime: capacity and reliability},'' \emph{IEEE Transactions on
  Information Theory}, vol.~53, no.~6, pp. 1983--2009, 2007.

\bibitem{journals/twc/LozanoP12}
A.~Lozano and D.~Porrat, ``{Non-peaky signals in wideband fading channels:
  Achievable bit rates and optimal bandwidth.}'' \emph{IEEE Transactions on
  Wireless Communications}, vol.~11, no.~1, pp. 246--257, 2012.

\bibitem{Ferrante2016}
G.~C. Ferrante, T.~Q.~S. Quek, and M.~Z. Win, ``{Revisiting the capacity of
  noncoherent fading channels in mmWave system},'' \emph{IEEE International
  Conference on Communications, ICC}, vol.~65, no.~8, pp. 3259--3275, 2016.

\bibitem{fgomezUnified}
F.~G{\'{o}}mez-Cuba, J.~Du, M.~M{\'{e}}dard, and E.~Erkip, ``{Unified capacity
  limit of non-coherent wideband fading channels},'' \emph{IEEE Transactions on
  Wireless Communications}, vol.~16, no.~1, pp. 43--57, 2017.

\bibitem{Du2017}
J.~Du and R.~A. Valenzuela, ``{How much spectrum is too much in millimeter wave
  wireless access},'' vol. 8716, no.~c, pp. 1--15, 2017.

\bibitem{journals/ftcit/TulinoV04}
A.~M. Tulino and S.~Verd{\'{u}}, ``{Random matrix theory and wireless
  communications.}'' \emph{Foundations and Trends in Communications and
  Information Theory}, vol.~1, no.~1, 2004.

\bibitem{marzetta2006much}
T.~L. Marzetta, ``{How much training is required for multiuser MIMO?}'' in
  \emph{IEEE 40th Asilomar Conference on Signals, Systems and Computers
  (ACSSC)}, 2006, pp. 359--363.

\bibitem{marzetta2013special}
T.~L. Marzetta, G.~Caire, M.~Debbah, I.~Chih-Lin, and S.~K. Mohammed,
  ``{Special issue on massive MIMO},'' \emph{Journal of Communications and
  Networks}, vol.~15, no.~4, pp. 333--337, 2013.

\bibitem{Bjornson2014a}
E.~Bj{\"{o}}rnson, E.~G. Larsson, and M.~Debbah, ``{Massive MIMO for maximal
  spectral efficiency: How many users and pilots should be allocated?}''
  \emph{IEEE Transactions on Wireless Communications}, vol.~15, no.~2, pp.
  1293--1308, 2016.

\bibitem{Bjornson2016}
E.~Bj{\"{o}}rnson, E.~G. Larsson, and T.~L. Marzetta, ``{Massive MIMO: ten
  myths and one critical question},'' \emph{IEEE Communications Magazine},
  vol.~54, no.~2, pp. 114--123, 2016.

\bibitem{zhengoptimal}
L.~Zheng and D.~Tse, ``{Optimal diversity-multiplexing trade-off in
  multi-antenna channels},'' in \emph{Proc. of the 39th Allerton Conference on
  Communication, Control and Computing}, 2001, pp. 835--844.

\bibitem{866662}
L.~Zheng and D.~N.~C. Tse, ``{Sphere packing in the Grassmann manifold: A
  geometric approach to the noncoherent multi-antenna channel},'' in \emph{2000
  IEEE International Symposium on Information Theory}, 2000, pp. 364--364.

\bibitem{hoydis2018}
E.~Bj{\"{o}}rnson, J.~Hoydis, and L.~Sanguinetti, ``{Massive MIMO has unlimited
  capacity},'' \emph{IEEE Transactions on Wireless Communications}, vol.~17,
  no.~1, pp. 574--590, 2018.

\bibitem{manolakos2014globecom}
A.~Manolakos, M.~Chowdhury, and A.~J. Goldsmith, ``{Constellation design in
  noncoherent massive SIMO systems},'' in \emph{IEEE Global Telecommunications
  Conference (GLOBECOM)}, 2014.

\bibitem{Chowdhury2014a}
M.~Chowdhury, A.~Manolakos, and A.~J. Goldsmith, ``{Design and performance of
  non coherent massive SIMO systems},'' in \emph{IEEE Annual Conference on
  Information Sciences and Systems (CISS)}, 2014, pp. 1--6.

\bibitem{7565517}
A.~Manolakos, M.~Chowdhury, and A.~Goldsmith, ``{Energy-based modulation for
  noncoherent massive SIMO systems},'' \emph{IEEE Transactions on Wireless
  Communications}, vol.~15, no.~11, pp. 7831--7846, 2016.

\bibitem{7404014}
M.~Chowdhury, A.~Manolakos, and A.~Goldsmith, ``{Scaling laws for noncoherent
  energy-based communications in the SIMO MAC},'' \emph{IEEE Transactions on
  Information Theory}, vol.~62, no.~4, pp. 1980--1992, 2016.

\bibitem{Knuth1976}
D.~E. Knuth, ``{Big Omicron and big Omega and big Theta},'' \emph{ACM SIGACT
  News}, vol.~8, no.~2, pp. 18--24, 1976.

\bibitem{Lozano2012}
A.~Lozano and N.~Jindal, ``{Are yesterday's information-theoretic fading models
  and performance metrics adequate for the analysis of today's wireless
  systems?}'' \emph{IEEE Communications Magazine}, vol.~50, no.~11, pp.
  210--217, 2012.

\bibitem{NRrelease15}
3GPP, ``{3rd Generation Partnership Project; Technical Specification Group
  Radio Access Network; NR; NR and NG-RAN Overall Description; Stage 2 (Release
  15)},'' \emph{3GPP TS 38 300}, no. Release 15, 2019.

\bibitem{Raghavan2007}
V.~Raghavan, G.~Hariharan, and A.~M. Sayeed, ``{Capacity of Sparse Multipath
  Channels in the Ultra-Wideband Regime},'' \emph{IEEE Journal of Selected
  Topics in Signal Processing}, vol.~1, no.~3, pp. 357--371, 2007.

\bibitem{cover2006elements}
T.~M. Cover and J.~a. Thomas, \emph{{Elements of Information Theory}}, ser.
  Wiley Series in Telecommunications.\hskip 1em plus 0.5em minus 0.4em\relax
  New York, USA: John Wiley {\&} Sons, Inc., 1991, vol.~6.

\end{thebibliography}
\end{document}